\documentclass{article}
\def\anonymous{0}

\usepackage{graphicx} 
\def\withcolors{1}
\def\withnotes{1}
\usepackage{ccanonne}
\usepackage{multicol}
\usepackage{multirow}
\usepackage{subcaption}
\usepackage{bbm}
\usepackage{braket}
\usepackage{algorithm}
\usepackage{algpseudocode}
\newcommand{\tr}{\operatorname{Tr}}

\DeclareMathOperator*{\argmin}{arg\,min}

\newtheorem{fact}[theorem]{Fact}
\allowdisplaybreaks
\newcolumntype{C}{>{\centering\arraybackslash}p{0.27\textwidth}}
\newcommand{\dims}{d}

\newcommand{\povmset}{{\mathfrak{M}}}
\newcommand{\nqubits}{{N}}

\newcommand{\pauliX}{{\sigma_X}}
\newcommand{\pauliY}{{\sigma_Y}}
\newcommand{\pauliZ}{{\sigma_Z}}
\newcommand{\pauliObsSet}{{\mathcal{P}}}

\newcommand{\opnorm}[1]{{\left\|#1\right\|}_{\text{op}}}
\newcommand{\tracenorm}[1]{{\left\|#1\right\|}_{1}}
\newcommand{\hsnorm}[1]{{\left\|#1\right\|}_{\text{HS}}}
\newcommand{\barDelta}{{\overline{\Delta}}}
\newcommand{\ptb}{{z}}
\newcommand{\ptbDistr}{{\mathcal{D}_{\ell,\cd}}}

\newcommand{\cd}{{c}}
\newcommand{\cop}{\kappa}

\newcommand{\x}{\mathbf{x}}
\newcommand{\out}{{x}}

\def\multiset#1#2{\ensuremath{\left(\kern-.3em\left(\genfrac{}{}{0pt}{}{#1}{#2}\right)\kern-.3em\right)}}

\newcommand{\ham}[2]{\operatorname{d}_{\rm Ham}(#1,#2)}

\newcommand{\qmm}{{\rho_{\text{mm}}}}

\newcommand{\HH}{\mathbb{H}}
\newcommand{\Herm}[1]{{\HH_{#1}}}

\newcommand{\qbit}[1]{|{#1}\rangle}
\newcommand{\qadjoint}[1]{\langle{#1}|}
\newcommand{\qproj}[1]{\qbit{#1}\qadjoint{#1}}
\newcommand{\qoutprod}[2]{\qbit{#1}\qadjoint{#2}}

\newcommand{\qdotprod}[2]{\langle#1|#2\rangle}
\newcommand{\hdotprod}[2]{\left\langle#1,#2\right\rangle}
\newcommand{\matdotprod}[3]{\langle#1|#2|#3\rangle}

\newcommand{\Var}{\text{Var}}
\newcommand{\eye}{\mathbb{I}}

\newcommand{\img}{\text{i}}

\newcommand{\rk}{{r}}
\newcommand{\VecOp}{\text{vec}}
\newcommand{\vvec}[1]{|#1\rangle\rangle}
\newcommand{\vadj}[1]{\langle\langle#1|}
\newcommand{\vvdotprod}[2]{\left\langle\left\langle#1|#2\right\rangle\right\rangle}

\newcommand{\outset}{{\mathcal{X}}}

\newcommand{\Luders}{\mathcal{H}}

\newcommand{\Choi}{{\mathcal{C}}}

\newcommand{\hbasis}{{\mathcal{V}}}
\newcommand{\qest}{{\hat{\rho}}}

\newcommand{\POVM}{\mathcal{M}}

\newcommand{\ObsPOVM}{\mathcal{N}}

\newcommand{\zest}{\hat{\ptb}}
\title{Pauli measurements are not optimal for single-copy tomography}
\ifnum\anonymous=1
    \author{Anonymous authors}
\else
    \author{
    \begin{tabular}[t]{C@{\extracolsep{6.5em}} C}
   Jayadev Acharya &Abhilash Dharmavarapu \\
 Cornell University & Cornell University\\ 
\small \texttt{acharya@cornell.edu} &\small \texttt{ad2255@cornell.edu} 
\end{tabular}
\vspace{2ex}\\
\begin{tabular}[t]{C@{\extracolsep{6.5em}} C}
    Yuhan Liu & Nengkun Yu \\
Rice University & Stony Brook University\\ 
\small \texttt{yuhan-liu@rice.edu} &\small \texttt{nengkun.yu@cs.stonybrook.edu} 
\end{tabular}
}
\fi

\begin{document}
\maketitle
\begin{abstract}
Quantum state tomography is a fundamental problem in quantum computing.
Given $n$ copies of an unknown $N$-qubit state $\rho\in\mathbb{C}^{d\times d},d=2^N$, the goal is to learn the state up to an accuracy $\varepsilon$ in trace distance, say with at least constant probability 0.99. We are interested in the copy complexity, the minimum number of copies of $\rho$ needed to fulfill the task.

As current quantum devices are physically limited, Pauli measurements have attracted significant attention due to their ease
of implementation. However, a large gap exists in the 
literature for tomography with Pauli measurements.
The best-known upper bound is $O(\frac{N\cdot 12^N}{\varepsilon^2})$, 
and no non-trivial lower bound is known besides the general single-copy lower bound of
$\Omega(\frac{8^N}{\varepsilon^2})$, achieved by hard-to-implement structured POVMs such as MUB, SIC-POVM, and uniform POVM.

We have made significant progress on this long-standing problem. We first prove a stronger upper bound of $O(\frac{10^{N}}{\varepsilon^2})$. To complement it, we also obtain a lower bound of $\Omega(\frac{9.118^N}{\varepsilon^2})$, which holds even with adaptivity. To our knowledge, this demonstrates the first known separation between Pauli measurements and structured POVMs. 

The new lower bound is a consequence of a novel framework for adaptive quantum state tomography with measurement constraints. 
The main advantage is that we can use measurement-dependent hard instances to prove tight lower bounds for Pauli measurements, 
while prior lower-bound techniques for tomography only work with measurement-independent constructions. 
Moreover, we connect the copy complexity lower bound of tomography to the eigenvalues of the measurement information channel, which governs the measurement’s capacity to distinguish between states. To demonstrate the generality of the new framework, we obtain tight bounds for adaptive quantum state tomography with $k$-outcome measurements, where we recover existing results and establish new ones.

\end{abstract}

\section{Introduction}

\textit{Quantum state tomography}~\cite{BBMR04,Keyl06,GJK08} is the problem of learning an unknown quantum state. It is a fundamental problem with important applications in quantum computing as we often need to learn about the state of a quantum device. 
Formally, we are given $\ns$ copies of an $\nqubits$-qubit quantum
state with density matrix $\rho\in \C^{\dims\times\dims}$ where $\dims=2^{\nqubits}$. We need to perform quantum
measurements on $\rho^{\otimes \ns}$ and obtain an estimate $\hat{\rho}$ close to $\rho$ under some error metric. In this work, we focus on the trace distance $\tracenorm{\qest-\rho}$, and we want to ensure that $\tracenorm{\qest-\rho}<\eps$ with probability at least 0.99. We want to characterize the copy/sample complexity, the minimum number of copies of $\rho$ needed for the task.

There are different types of measurements we can apply. The most general is  \emph{entangled} or \emph{joint} measurement, where one can arbitrarily apply any measurement to $\rho^{\otimes \ns}$.
In \cite{HaahHJWY17,ODonnellW16,ODonnellW17}, the authors showed that the worst-case sample complexity is $n =\bigTheta{\frac{4^\nqubits}{\eps^2}}$. 
While such measurement is powerful, it is difficult to implement on near-term quantum computers as it requires a large and coherent quantum memory.  
This leads to a line of work studying \emph{single-copy} measurements~\cite{Flammia_2012,KRT14,HaahHJWY17,chen2023does}, where we apply a separate measurement for each copy of $\rho$. 
The measurements can be chosen \emph{non-adaptively}, where all $\ns$ measurements are decided simultaneously, or \emph{adaptively}, where the copies are measured sequentially, and each measurement can be chosen based on previous outcomes. 
With non-adaptive measurements, \cite{KRT14} showed that 
tomography is possible using $\ns = O(8^\nqubits / \eps^2)$ copies, and later was shown to be optimal by \cite{HaahHJWY17}. \cite{chen2023does} further showed that adaptivity does not help. 

However, since each copy is an $\nqubits$-qubit system, single-copy measurement potentially requires entanglement over $\nqubits$ qubits. 
In particular, the optimal single-copy measurements \cite{KRT14,guctua2020fast} are highly structured POVMs (including SIC-POVM, maximal MUB, uniform POVM) that are difficult to implement in practice especially when $N$ is large.
Thus, it is important to study single-copy tomography under measurement constraints. 
As a canonical example, \emph{Pauli measurements} have attracted significant attention due to their ease of implementation. 
It only involves measuring each qubit in the eigenbasis of one of the three $2\times 2$ Pauli operators $\pauliX,\pauliY,\pauliZ$, which is arguably one of the most experiment-friendly measurements. 
Although Pauli measurement is fundamental in quantum physics, there is a large gap between the upper and lower bounds of its copy complexity. 
The best-known upper bound is $\ns=\bigO{\nqubits\frac{12^{\nqubits}}{\eps^2}}$ \cite{guctua2020fast}, and no better lower bound is known besides the single-copy lower bound of $\bigOmega{\frac
{8^{\nqubits}}{\eps^2}}$. On the other hand, empirical evidence \cite{stricker2022pauliSIC} suggests that SIC-POVM is better than Pauli measurements.


\subsection{Results}
Our work makes notable progress in closing the long-standing gap for Pauli tomography.
We prove the first non-trivial lower bound showing that Pauli measurements cannot achieve the optimal sample complexity for single-copy tomography.  
\begin{theorem}
Using Pauli measurements, learning an unknown $\nqubits$-qubit state $\rho$ up to trace distance $\eps$  with probability at least 0.99 requires at least
\[
\ns = \bigOmega{\frac{9.118^{\nqubits}}{\eps^2}}
\]
copies of $\rho$.
This lower bound holds even when the measurements are chosen adaptively.
\label{thm:pauli-lower}
\end{theorem}
Since highly structured POVM such as SIC-POVM and maximal MUB achieves the $O(8^{\nqubits}/\eps^2)$ sample complexity, we formally show a separation between Pauli measurements and these structured POVMs. This is consistent with experimental observations~\cite{stricker2022pauliSIC}.

We also design a new algorithm that improves the current upper bound,
\begin{theorem}
 {Quantum state tomography} can be solved by Pauli measurements  using 
 $$\ns = \bigO{\frac{10^N\log\frac{1}{\delta}}{\eps^2}}$$ copies of $\rho$ with success with probability at least $1-\delta$.
 \label{thm:pauli-upper}
 \end{theorem}

Thus, we significantly reduce the existing gap between upper and lower bounds for Pauli tomography. We conjecture that the $10^{\nqubits}/\eps^2$ upper bound is tight.

Our main technical contribution is a new lower-bound framework for adaptive quantum tomography under measurement constraints. The constraint is generally described by a set of POVMs $\povmset$, which is the set of allowed measurements that we can apply to each copy. In our problem, $\povmset$ is the set of Pauli measurements.

Compared to the unrestricted case, there are two challenges to proving tight lower bounds. 
The first challenge is that with constraint $\povmset$, some states might be harder to learn for measurements in $\povmset$ than others. 
Thus, we need to design a measurement-dependent hard instance to capture the limitations of $\povmset$. The second challenge is to quantitatively evaluate the effect of measurement constraints on sample complexity. 
Intuitively, there should be some indicator of ``measurement capability'' that controls the hardness of learning. 

To our knowledge, existing lower-bound techniques for tomography cannot address these challenges. The hard case constructions in many works \cite{ODonnellW16, HaahHJWY17,chen2023does,lowe2022lower} are measurement-independent. Moreover, the analysis is either oblivious to measurement constraints \cite{ODonnellW16, HaahHJWY17,chen2023does}, or only applies to some specific constraint such as Pauli observables~\cite{Flammia_2012} and constant-outcome measurements~\cite{lowe2022lower}.

We address both challenges and develop a general framework using the \emph{measurement information channel},

\begin{definition}
\label{def:mic}
    Let $\POVM=\{M_x\}_{x}$ be a POVM. The \emph{measurement information channel (MIC)} $\Luders_{\POVM}:\C^{\dims\times\dims}\mapsto\C^{\dims\times\dims}$ and its matrix representation $\Choi_{\POVM}$ are defined as
\begin{equation}
    \Luders_{\POVM}(A)\eqdef\sum_{x}M_x\frac{\Tr[M_xA]}{\Tr[M_x]}, \quad\Choi_{\POVM}\eqdef \sum_{x}\frac{\vvec{M_x}\vadj{M_x}}{\Tr[M_x]} \in \C^{\dims^2\times\dims^2}.
\end{equation}
where $\vvec{M_x}=\VecOp(M_x)$ and $\vadj{M_x}=\VecOp(M_x)^{\dagger}$.
\end{definition}
The channel maps a quantum state to another quantum state. 
Intuitively, it characterizes the similarity of the outcome distributions after applying $\POVM$ to $\rho$ and the maximally mixed state $\qmm$. 
The ability to distinguish between different states is described by the eigenvalues of the channel. The eigenvectors (which are matrices in this case) with small eigenvalues indicate the directions that are hard to distinguish for the measurement. 

The MIC helps us to address both challenges. We design hard instances based on the ``weak'' directions of MIC of measurements in $\povmset$, 
and the sample complexity would depend on the eigenvalues of MIC.
Our framework not only works for Pauli measurements but can also be applied to arbitrary measurement constraints $\povmset$. 
To demonstrate the generality of our approach, we prove tight sample complexity bound for a natural family of $\ab$-outcome measurements. 
\begin{theorem} \label{thm:nearly-tight-finite-out}
    The worst-case copy complexity of single-copy tomography with $\ab$-outcome measurements is
    \[
    \ns = \tildeTheta{\frac{\dims^4}{\eps^2\min\{\ab, \dims\}}}.
    \]
    The lower bound holds for adaptive measurements, and the upper bound is achieved by non-adaptive ones.
\end{theorem}
Previously, the worst-case bound was known only for constant $\ab$~\cite{lowe2022lower,Flammia_2012}\footnote{\cite{lowe2022lower} proved the lower bound for non-adaptive measurements. Their adaptive lower bound only applies to finite constraint set $\povmset$.} and $\ab\ge\dims$~\cite{HaahHJWY17, chen2023does,guctua2020fast}. We not only recover their results but establish a complete dependence on $\ab$. We give a detailed discussion of our framework in \cref{sec:techniques}.

\subsection{Related works}
\paragraph{Quantum state tomography} We make additional remarks about previously mentioned works and discuss other works in this regime.

Many works study the tomography of low-rank states. 
For $\rho$ with rank $\rk$, it is shown that $\ns = \tildeTheta{\frac{\dims \rk}{\eps^2}}$ is necessary and sufficient~\cite{HaahHJWY17,ODonnellW16} with entangled measurement. 
For single-copy measurements, the sample complexity is $\ns=\bigTheta{\frac{\dims \rk^2}{\eps^2}}$ for non-adaptive measurements, but whether it is tight for adaptive ones is unknown. 

For Pauli measurements, \cite{guctua2020fast} showed that $\ns=\bigO{\frac{\nqubits \cdot 3^\nqubits\cdot \rk^2}{\eps^2}}$ is sufficient. 
Random Pauli measurements offer distinct advantages \cite{Elben_2022} and have been effectively applied in quantum process tomography for shallow quantum circuits \cite{yu2023learningmarginalssuffices,Huang_2024}. 

Some work~\cite{compressed,Flammia_2012} considers Pauli observables, a special class of 2-outcome Pauli measurements. The sample complexity for rank-$\rk$ state tomography is ${\tilde{{\Theta}}}(\frac{\dims^2 r^2}{\eps^2})$ using non-adaptive measurements~\cite{Flammia_2012}. \cite{lowe2022lower} showed that the lower bound also holds for adaptively chosen constant-outcome measurements\footnote{More precisely when adaptively chosen from a finite set of measurements with size at most $\exp(O(\dims))$.}. \cite{cai2016optimal} derived near-optimal error rates for Hilbert-Schmidt and operator-norm distance. However, they require that the state is sparse in the expectation values of Pauli observables, nor did they prove a lower bound for the trace distance.

Recently, \cite{Chen0L24memory} studied the case when one can perform entangled measurement over $t>1$ copies at a time, which interpolates between single-copy and fully entangled measurements. Apart from trace distance, other metrics were considered, such as fidelity~\cite{HaahHJWY17,chen2023does, Yuen_2023}, quantum relative entropy~\cite{flamian2023tomography}, and Bures $\chi^2$-divergence~\cite{flamian2023tomography}. Extending our work to low-rank states and other error metrics is an interesting future direction.

\paragraph{Other quantum state inference problems} Quantum state testing/certification \cite{ODonnellW15,BadescuO019} is a closely related problem, where the goal is to test whether an unknown state $\rho$ is equal or far from a target state $\sigma$. The problem has also been considered under entangled~\cite{ODonnellW15,BadescuO019}, single-copy measurements~\cite{BubeckC020,Chen0HL22,liu2024role}, and Pauli measurements~\cite{Yu2023almost}. A concurrent work~\cite{liu2024restricted} considers single-copy testing under measurement restrictions, but their technique only applies to non-adaptive measurements.

In practice, we are often interested in partial information about the state rather than a full-state description. \cite{Cotler_2020, Garc_a_P_rez_2020, evans2019scalable} studied \textit{quantum overlapping tomography}, where the goal is to output the classical description of all $k$-qubit reduced
density matrices of an $n$-qubit system. The algorithms are based on Pauli measurements, demonstrating their wide applicability. 
Shadow tomography~\cite{Aaronson20, huang2020predicting, ChenCH021,chen2024pauli} aims to learn the expectation value of a finite set of observables. In particular, \cite{chen2024pauli} studied Pauli shadow tomography with various measurement constraints such as measurements with finite quantum memory and Clifford measurements. It would be an interesting future work to establish a formal connection between their framework and our method.

\paragraph{Distributed distribution estimation}

Quantum tomography can be thought of the quantum analogue of the classical problem of distribution estimation. Given samples from an unknown distribution $p$, the goal is to output an estimate $\hat p$ such that $d(p,\hat p)<\varepsilon$ for some distance $d$. The problem has a long history and has been studied under several different settings. 

The problem of single copy tomography is in spirit similar to the problems of distributed estimation of distributions under information constraints. 
In it, i.i.d. samples $X_1, \ldots, X_n$ from the unknown $p$ are distributed across $n$ users, who are constrained in how much information about their sample they can send. 
Well-studied information constraints include communication constraints (where each user has to compress their sample using at most $b$ bits), or privacy constraints (where each user has to add noise to their sample to preserve privacy). Several problems in distributed estimation of distribution for both discrete, continuous as well as high dimensional distributions have been studied in the past several years~\cite{duchi2013local, barnes2019lower, AcharyaCT19, acharya2020distributed}.

\section{Our techniques}
\label{sec:techniques}
\subsection{Lower bound ideas through a simple example}
Our main contribution is a novel technique to prove single-copy tomography lower bounds with measurement restrictions. Before we delve into the details, we use a simple example to illustrate why we need new ideas for the problem.

Let's say that we are only allowed to use the computational basis measurement $\{\qproj{x}\}_{x=0}^{\dims-1}$ for each copy. It is impossible to perform quantum tomography under this constraint: one can easily observe that for both the maximally mixed state $\qmm\eqdef\eye_\dims/\dims$ and the state 
\[
\qbit{\psi}=\frac{1}{\sqrt{\dims}}\sum_{x=0}^{\dims-1}\qbit{x},
\]
the measurement outcomes would be a uniform distribution over $\{0, \ldots, \dims-1\}$. Yet, the trace distance between the two states is close to 1. Thus, we cannot even distinguish two states that are nearly as far away as they can be, let alone learning any given state up to an arbitrary accuracy $\eps$.

An immediate lesson is that when the constraint set $\povmset$ is too restricted, nature would be able to design some states that are particularly hard to distinguish for measurements in this set $\povmset$. In this example, the two states $\qbit{\psi}$ and $\qmm$ are precisely chosen based on the measurement $\{\qproj{x}\}_{x=0}^{\dims-1}$. 

Note that Pauli measurements only consist of $3^{\nqubits}$ basis measurements. 
It is a fairly small set compared to the dimensionality of quantum states which is $\dims^2=4^{\nqubits}$. Furthermore, it does not have nice geometric properties of maximal MUB~\cite{klappenecker2005mutually} and SIC-POVM
~\cite{zauner1999grundzuge}.
Thus, we expect that the lower bound instance for Pauli tomography also needs to be \emph{measurement-dependent}.
However, to our knowledge, the constructions in existing works on quantum tomography are predominantly measurement-independent.
For example, \cite{chen2023does} uses Gaussian orthogonal ensembles which informally speaking apply independent Gaussian perturbations to each coordinate of the maximally mixed state. \cite{HaahHJWY17} constructs a packing set based on Haar-random unitary transformations. 
Thus, the techniques in these works are likely not optimal for Pauli measurements.

It was fairly easy to design two states that completely fool the computational basis measurement, which implies that tomography with just the computation basis is impossible.
For other measurement constraints such as Pauli measurements, we need a systematic approach to (1) design a specific hard case instance and (2) analyze the effect of the constraint on tomography. It turns out that the \emph{measurement information channel (MIC)} helps us to achieve both objectives. We illustrate the role of MIC using the computational basis example. From \cref{def:mic}, the MIC of $\POVM=\{\qproj{x}\}_{x=0}^{\dims-1}$ is
\[
\Luders_{\POVM}(\cdot) = \sum_{x=0}^{\dims-1} \qproj{x}(\cdot)\qproj{x}.
\]
The outputs of $\Luders_{\POVM}$ on $\qmm$ and $\qproj{\psi}$ are identical,
\[
\Luders_{\POVM}(\qmm)=\frac{1}{\dims}\sum_{x=0}^{\dims-1} \qproj{x}\eye_\dims\qproj{x}=\frac{1}{\dims}\sum_{x=0}^{\dims-1} \qproj{x},\quad 
\Luders_{\POVM}(\qproj{\psi})=\sum_{x=0}^{\dims-1} \qbit{x}\qdotprod{x}{\psi}\qdotprod{\psi}{x}\qadjoint{x}=\frac{1}{\dims}\sum_{x=0}^{\dims-1} \qproj{x}.
\]
Let $\Delta=\qmm-\qproj{\psi}$. Equivalently we have $\Luders_{\POVM}(\Delta)=0$, or $\Delta$ lies in the 0-eigenspace of $\Luders_{\POVM}$. Therefore, to construct the hard instance, we can perturb the reference state (normally $\qmm$) along directions where the MIC has small eigenvalues. Furthermore, the spectrum of MIC in the constraint set $\povmset$ determines the hardness of tomography.

The intuition might appear similar to how testing with fixed measurement is disadvantageous to randomized ones~\cite{liu2024role}. However, their work does not consider measurement restrictions for each copy. Moreover, tomography is a harder problem than testing which requires a different analysis. Furthermore, we allow adaptive measurements, which are much more complicated than fixed measurements where the measurement outcomes are independent. Thus, it is unclear how their arguments can extend to our problem.

\subsection{The lower bound construction} 
Informally, our construction adds independent binary perturbations to $\qmm$ along different directions,
\begin{equation}
\label{equ:quantum-construction-informal}
    \sigma_z=\qmm + \frac{\dst}{\sqrt{\dims}}\cdot\frac{\cd}{\dims}\sum_{i=1}^{\ell} z_iV_i, \quad\ell=\frac{\dims^2}{2},
\end{equation}
where $\{V_i\}_{i=1}^{\dims^2-1}$ are $\dims^2-1$ orthonormal trace-0 Hermitian matrices which satisfy $\Tr[V_iV_j]=\indic{i=j}$, $z=(z_1, \ldots, z_{\dims^2/2})$ is drawn uniformly from $\{-1, 1\}^{\dims^2/2}$, and $\cd$ is an absolute constant. 

The same construction was used for quantum state testing by \cite{liu2024role}. We can argue that with high probability over $z$, 
\eqref{equ:quantum-construction-informal} yields a valid quantum state that is $\eps$ far from the maximally mixed state $\qmm$. 
\begin{theorem}[Valid construction, informal]
\label{thm:valid-construction-informal}
    Let $\eps\le 1/200$, and $c$ be a suitably chosen absolute constant. Let $z\sim \{-1,1\}^{\dims^2/2}$ uniformly, then with probability at least $1-\exp(-\dims)$, $\sigma_z$ in \eqref{equ:quantum-construction-informal} is a valid quantum state and $\tracenorm{\sigma_z-\qmm}>\eps$. 
\end{theorem}
In the rare event that $\sigma_z$ is not a valid state (i.e., not p.s.d.), we shrink the perturbation $\Delta_z=\sigma_z-\qmm$ so that it has a maximum eigenvalue of at most $1/(2\dims)$. The formal definition is presented in \cref{def:perturbation}.

The main advantage of this construction is that it gives us the freedom to choose directions $V_1, \ldots, V_{\dims^2/2}$ that are least sensitive to the given measurement constraint. Next, we discuss the choice of these directions for Pauli measurements.

\paragraph{Construction for Pauli measurements} We choose these directions to be the (normalized) \emph{Pauli observables} with the largest weights. 
In short, a Pauli observable $P$ belongs to $\pauliObsSet=\{\pauliX, \pauliY, \pauliZ, \eye_2\}^{\otimes \nqubits}\setminus\{\eye_\dims\}$\footnote{Some literature also include $\eye_\dims$ as a Pauli observable} where 
\begin{equation}
    \pauliX = \begin{bmatrix}
    0 & 1 \\
    1 & 0
\end{bmatrix},\quad
\pauliY = \begin{bmatrix}
    0 & -i \\
    i & 0
\end{bmatrix},\quad
\pauliZ = \begin{bmatrix}
    1 & 0\\
    0 & -1
\end{bmatrix}.
\label{equ:pauli-ops}
\end{equation}
The \emph{weight} of $P$ is the number of non-identity components ($\pauliX,\pauliY,\pauliZ$) it contains. 
An important property is that the Pauli observables and $\eye_\dims$ forms an orthogonal basis for quantum states, and thus any state $\rho$ can be represented as
\begin{equation}
    \rho = \frac{\eye_\dims}{\dims} +\sum_{P\in\pauliObsSet}\alpha_PP, \quad\alpha_P=\frac{\Tr[\rho P]}{\dims}.
    \label{equ:pauli-decomposition}
\end{equation}

For each copy of the $\nqubits$-qubit state, a Pauli measurement $\POVM$ measures each qubit in the eigenbasis of either $\pauliX,\pauliY,\pauliZ$. The outcome is a $\{-1, 1\}^{\nqubits}$ binary string which reveals information about the coefficient $\alpha_P$ of all Pauli observables $P$ whose non-identity components match the choice of $\pauliX,\pauliY,\pauliZ$ for the corresponding qubit. Thus, the coefficients of $P$ with larger weights are harder to learn. 

As an example, if $P=\sigma_X^{\otimes \nqubits}$, then the only way to learn about $\alpha_P$ is to measure all qubits in the eigenbasis of $\sigma_X$. On the other hand, to learn about $\sigma_X\otimes \eye_{\dims/2}$, the only requirement is to measure the first qubit in the eigenbasis of $\sigma_X$, and we can choose any of the three choices for other qubits. In general, for a Pauli observable $P$ with weight $w$, there are $3^{\nqubits-w}$ Pauli measurements that can learn information about $\alpha_P$.

\paragraph{The role of MIC} It turns out that the measurement information channel of Pauli measurements precisely formalizes our previous intuition. Using \cref{def:mic} and the definition of Pauli measurements, we have the following result,
\begin{lemma}[MIC of Pauli measurement, informal]
\label{lem:mic-pauli-informal}
    Let $\Luders_{\POVM}$ be the measurement information channel of a Pauli measurement $\POVM$, then for all Pauli observable $P$
    \[
    \Luders(P)=P\indic{\text{The non-identity components of $P$ match the choice of basis in $\POVM$}}.
    \]
\end{lemma}
This is consistent with the fact that Pauli measurements only reveal information for Pauli observables with a matching choice of $\pauliX, \pauliY, \pauliZ$. 

To see how the eigenvalues of MIC characterize the ability of Pauli measurements, we consider a POVM $\ObsPOVM$ defined by the uniform ensemble of all $3^{\nqubits}$ Pauli measurements. In other words, we uniformly sample a Pauli measurement and observe the outcome. Together with the choice of the measurement, $\ObsPOVM$ is a POVM with $3^{\nqubits}\cdot 2^{\nqubits}=6^{\nqubits}$ outcomes. From \cref{def:mic}, the MIC of $\ObsPOVM$ is simply the linear combination of the MIC of all Pauli measurements,
\[
\Luders_{\ObsPOVM}(\cdot) = \frac{1}{3^{\nqubits}}\sum_{\POVM \text{ Pauli}}\Luders_{\POVM}(\cdot).
\]
Due to linearity, all Pauli observables $P$ are also the eigenvectors of $\Luders_{\ObsPOVM}$. Suppose $P$ has a weight of $w$, then using \cref{lem:mic-pauli-informal}, its eigenvalue for $\Luders_{\ObsPOVM}$ is,
\[
\frac{1}{3^{\nqubits}}\sum_{\POVM \textbf{ Pauli}}\indic{\text{The non-identity components of $P$ match the choice of basis in $\POVM$}}=\frac{3^{\nqubits-w}}{3^{\nqubits}}=3^{-w}.
\]
The first step is precisely because there are $3^{\nqubits-w}$ Pauli measurements that match the non-identity components of $P$. Thus $P$ with a large weight has a smaller eigenvalue,  meaning that over uniform draw of Pauli measurements, we learn less about large-weight $P$ ``on average''. This is consistent with the previous discussion that large-weight Pauli observables are hard to learn for Pauli measurements.

\subsection{Assouad's method: Hamming separation for trace distance}
The packing argument is popular among previous works \cite{HaahHJWY17, ODonnellW16} which constructs a finite set of states such that the pair-wise trace distance is $\Omega(\eps)$. 
Thus, any learning algorithm must be able to correctly identify the state chosen by nature from the packing set. 
From this Holevo's theorem can be applied. However, it is not straightforward to construct a packing set that adjusts to the measurement constraint. 

Our lower bound is based on Assouad's method~\cite{Assouad83,yu1997assouad}, which reduces the learning problem to a multiple binary hypothesis testing problem. It has been extensively used for many parametric estimation problems \cite{duchi2013local,ACLST22iiuic}. The method is more suitable for our construction in \eqref{equ:quantum-construction-informal}.

Let $L(\cdot, \cdot)$ be an error metric between quantum states.
The main argument is that if an algorithm can learn any state with a small error in terms of $L$, then given a randomly sampled $\sigma_z$, the algorithm should be able to obtain an estimate $\hat{z}\in\{-1,1\}^{\dims^2/2}$ that matches most coordinates of $z$. Traditionally \cite{yu1997assouad,duchi2013local}, this relies on a $2\tau$-Hamming separation for the error metric $L$,
\[
L(\sigma_z,\sigma_{z'})\ge 2\tau \ham{z}{z'}=2\tau\sum_{i=1}^{\dims^2/2}\indic{z_i\ne z_i'}.
\]
Given this relation, a small loss $L$ implies $\ham{z}{z'}$ must be small. We then need to compute the separation parameter $\tau$, which is easy if $L$ and $z$ have a coordinate-wise relation. This is often true for classical distribution estimation~\cite{duchi2013local,ACLST22iiuic}, where $L$ is often the $\ell_p$ norm between two distributions. In quantum tomography, if $L$ is the Hilbert-Schmidt distance $\hsnorm{\sigma_z-\sigma_{z'}}$, the distance can also be written in terms of the coordinates of $z,z'$ since $V_i$'s are orthogonal. \cite{cai2016optimal} obtained the lower bound for the Hilbert-Schmidt distance using this approach.

However, the trace distance does not have a nice geometry like the Hilbert-Schmidt norm or vector $\ell_p$ norms that yields a direct relation between $\tracenorm{\sigma_z-\sigma_{z'}}$ and each coordinate of $z,z'$. 
Partly for this reason, \cite{cai2016optimal} did not obtain a lower bound for trace distance and suggested that a new approach might be needed.

Instead of trying to prove a general coordinate-wise relation for trace distance, we show a Hamming separation only for the ``good'' $z\in\{-1,1\}^{\dims^2/2}$ such that $\sigma_z$ is a valid quantum state. 
This is sufficient for our purpose since  according to \cref{thm:valid-construction-informal} an overwhelming fraction of $z$'s are ``good''.
\begin{lemma}[Trace distance Hamming separation, informal] \label{lemma:hamm-separation-informal}
   Let $z\in\{-1, 1\}^{\dims^2/2}$ and $c'$ be an absolute constant. If $\sigma_z$ defined in \eqref{equ:quantum-construction-informal} is a valid quantum state, then for all  $z' \in \left\{-1,1\right\}^{\dims^2/2}$, 
   \begin{equation}
       \tracenorm{\sigma_z - \sigma_{z'}} \geq \frac{c' \eps}{\dims^2} \ham{z}{z'}.
   \end{equation}
\end{lemma}
\begin{proof}[Proof sketch]
    The idea is that when $\sigma_z$ is ``good'', then the perturbation $\Delta_z=\sigma_z-\qmm$ has an operator norm at most $O(\eps/\dims)$. Then we use the duality between the trace norm and operator norm,
\begin{lemma}[Duality between trace and operator norms] \label{lemma:trace-norm-dual}
    Let $A\in \C^{\dims\times\dims}$, then
    \[
    \tracenorm{A}=\sup_{B\in \C^{\dims\times\dims}:\opnorm{B}\le 1}|\Tr[B^{\dagger}A]|.
    \]
\end{lemma}
We set $A=\sigma_z-\sigma_{z'}$ and $B=\Delta_z/\opnorm{\Delta_z}$. For simplicity of presentation assume that $\sigma_{z'}$ is also a valid quantum state, then $A=\frac{2c\eps}{\dims\sqrt{\dims}}\sum_{i}\indic{z_i\ne z_i'}V_i$. Note that $\Delta_z=\frac{c\eps}{\dims\sqrt{\dims}}\sum_{i}V_i$. Then by duality,
\[
\tracenorm{A}\ge \frac{\Tr[\Delta_z A]}{\opnorm{\Delta_z}}=\bigOmega{\frac{\dims}{\eps}}\frac{2c^2\eps^2}{\dims^3}\sum_{i}\indic{z_i\ne z_i'}=\bigOmega{\frac{\eps}{\dims^2}}\ham{z}{z'}.
\]
The second step uses that $V_i$'s are orthonormal and thus $\Tr[V_iV_j]=\indic{i=j}$. 
We note that we can also prove the Hamming separation when $\sigma_{z'}$ is not a valid state and $\Delta_{z'}$ needs to be normalized. The formal lemma statement and proof are in \cref{lemma:hamm-packing}.
\end{proof}

Using \cref{lemma:hamm-separation-informal}, we can argue that a tomography algorithm must be able to guess at least 0.59 fraction of the $z_i$'s correctly,
\begin{proposition}
\label{prop:tomography-guess}
    Let $z\sim \{-1,1\}^{\ell}$ be uniform. Given $\ns$ copies of $\sigma_z$, a tomography algorithm with accuracy $\eps$ in trace distance can obtain a guess $\hat{z}\in \{-1,1\}^{\ell}$ such that
    \[
    \frac{1}{\ell}\sum_{i=1}^\ell\probaOf{z_i\ne\hat{z}_i}\le 0.41.
    \]
\end{proposition}

\subsection{Handling adaptivity via average mutual information}
Let $x_1, \ldots, x_{\ns}$ be the measurement outcomes, and denote $x^t=(x_1, \ldots x_t)$. To complete Assouad's argument, we need to analyze the outcome distributions when the $i$th coordinate is fixed $z_i=+1$ or $z_i=-1$ while other $z_j$ are still chosen uniformly. Denote these distributions as $\p_{+i}^{x^{\ns}}$ and  $\p_{+i}^{x^{\ns}}$ respectively. Using Le Cam's method \cite{LeCam73,yu1997assouad}, the total variation distance must be large to guess $z_i$ correctly,
\[
\Pr_{z_i\sim \{-1,1\}}[z_i\ne \hat{z}_i(x^\ns)]\ge \frac{1}{2}\Paren{1-\totalvardist{\p_{+i}^{x^{\ns}}}{\p_{-i}^{x^{\ns}}}}.
\]
Here $\hat{z}_i$ is an estimator that guesses $z_i$, which is produced by the tomography algorithm in our case. 
Combining with \cref{lemma:hamm-separation-informal}, it is sufficient to upper bound each $\totalvardist{\p_{+i}^{x^{\ns}}}{\p_{-i}^{x^{\ns}}}$. 

However, the total variation distance is hard to compute, especially given that $\p_{+i}^{x^{\ns}}$ are $\p_{-i}^{x^{\ns}}$ complicated mixture distributions.
Furthermore, the dependence between the outcomes $x_1, \ldots, x_{\ns}$ due to adaptivity poses another great challenge. Instead, we use the \emph{average mutual information} \cite{ACLST22iiuic} between the outcomes $x^\ns$ and $z_i$ as a bridge,
$\frac{1}{\ell}\sum_{i=1}^{\ell}\mutualinfo{z_i}{x^\ns}, \ell=\frac{\dims^2}{2}$.

First, when $z\sim \{-1, 1\}^{\ell}$, we can easily relate this quantity to the average error probability of guessing the $z_i$'s using \cite[Lemma 10]{ACLST22iiuic},
\begin{equation}
    \frac{1}{\ell}\sum_{i=1}^{\ell}\mutualinfo{z_i}{x^\ns}\ge 1-h\Paren{\frac{1}{\ell}\sum_{i=1}^\ell\probaOf{z_i\ne\hat{z}_i}}, \;h(p)=-p\log p-(1-p)\log(1-p).
    \label{equ:avg-MI-lower}
\end{equation}

From \cref{prop:tomography-guess}, and that $h$ is increasing in $[0, 1]$, the average mutual information must be lower bounded by a constant $1-h(0.41)$.

It remains to upper bound the average mutual information. Mutual information can be expressed using the KL-divergence, which enjoys a chain rule that helps us to analyze the distribution of each outcome $x_i$ separately even though it may depend on previous outcomes $x^{i-1}$. By further upper-bounding KL-divergence using chi-square divergence, we obtain an upper bound in terms of the measurement information channel. The formal statement is in \cref{thm:avg-MI-upper} and we provide an informal one below.

\begin{theorem}[Average mutual information bound, informal]
\label{thm:avg-MI-upper-informal}
    Let $\ptb\sim\{-1,1\}^{\ell}$ and $\sigma_z$ defined in \eqref{equ:quantum-construction-informal}, and $\out^\ns$ be measurement outcomes. Then,
    \begin{align}
         \frac{1}{\ell}\sum_{i=1}^{\ell}\mutualinfo{\ptb_i}{\out^\ns}&\le  \frac{c_1\ns\eps^2}{\ell^2} \sup_{\POVM\in {\povmset}}\sum_{i=1}^{\ell} \vadj{V_i} \Choi_{\POVM} \vvec{V_i}\label{equ:avg-MI-partial-informal},
    \end{align}
    where $\Choi_{\POVM}$ is the matrix representation of  $\POVM$'s MIC and $c_1$ is an absolute constant.
\end{theorem}

\paragraph{Proof of Pauli tomography lower bound} Recall that we choose $V_1, \ldots, V_{\dims^2/2}$ as the normalized Pauli observables with the largest weights. This roughly consists of all Pauli observables with weights at least $\frac{3\nqubits}{4}$. Using \cref{lem:mic-pauli-informal},  for all Pauli measurement $\POVM$, we have $\vadj{V_i} \Choi_{\POVM} \vvec{V_i}=1$ if the non-identity components of $V_i$ match with $\POVM$ and 0 otherwise. Thus,
\[
\sum_{i=1}^{\ell} \vadj{V_i} \Choi_{\POVM} \vvec{V_i}=\sum_{w=3\nqubits/4}^{\nqubits}{\nqubits\choose w}=\sum_{w=0}^{\nqubits/4}{\nqubits\choose w}= \bigO{2^{\nqubits h(1/4)}}.
\]
The first equality is because there are ${\nqubits\choose w}$  $V_i$'s of weight $w$ with matching non-identity components. 
The final step is due to Stirling's approximation (see \cref{lem:sum-binomial-coef}). 
Due to linearity, the above expression holds for all convex combinations of $\Choi_{\POVM}$. Combining with \eqref{equ:avg-MI-lower} and \eqref{equ:avg-MI-partial-informal}, and noting that $\ell=\dims^2/2$,
\[
\ns =\bigOmega{\frac{2^{\nqubits(4-h(1/4))}}{\eps^2}}.
\]
Finally note that $2^{4-h(1/4)}\ge 9.118$. This completes the lower bound proof.

\paragraph{Plug-and-play lower bound} The summation in \cref{equ:avg-MI-partial-informal} can be further bounded as $\sum_{i=1}^{\ell} \vadj{V_i} \Choi_{\POVM} \vvec{V_i}\le \tracenorm{\Choi_{\POVM}}=\tracenorm{\Luders_{\POVM}}$.
Combining with \eqref{equ:avg-MI-lower}, we have a more convenient result,
\begin{corollary}
    Let $\eps\le 1/200$ and $\povmset$ be the set of allowed measurements for each copy. Then the sample complexity of adaptive tomography with constraint $\povmset$ is
    \[
    \ns = \bigOmega{\frac{\dims^4}{\eps^2\sup_{\POVM\in {\povmset}}\tracenorm{\Luders_{\POVM}}}}.
    \]
\end{corollary}
Thus we provide a plug-and-play adaptive tomography lower bound for all measurement constraints $\povmset$ without going through the complications of adaptivity. We demonstrate it for finite-outcome measurements,
\begin{corollary} \label{cor:finite-lower}
    By \cite[Lemma 2.3]{liu2024restricted}, $\tracenorm{\Luders_{\POVM}}\le \min\{\ab, \dims\}$ for $\POVM$ with at most $\ab$ outcomes. Thus, the sample complexity lower bound for adaptive tomography with $\ab$-outcome measurements is
    \[
    \ns = \bigOmega{\frac{\dims^4}{\eps^2\min\{\ab, \dims\}}}.
    \]
\end{corollary}
This recovers the lower bound in \cite{lowe2022lower} for constant-outcome measurements (in fact improves it by a $\log\dims$ factor) and the single-copy adaptive lower bound in \cite{chen2023does}. Furthermore, we show in \cref{sec:finite-upper} that the lower bound is tight up to log factors for general $\ab$.

\subsection{Refined upper bound analysis}
Like the previous work of \cite{guctua2020fast}, we use a non-adaptive scheme where we apply each Pauli measurement to the same number of copies $m\eqdef\ns/3^{\nqubits}$. Each Pauli measurement contributes to the empirical estimate $\hat{\alpha}_P$ of $\alpha_P$ in \eqref{equ:pauli-decomposition} where $P$ has matching non-identity components. Our estimator is given by 
\[
\hat{\rho}=\frac{\eye_\dims}{\dims}+\sum_{P}\hat{\alpha_P}P.
\]
Using the fact that a weight-$w$ Pauli observable can be learned with $3^{\nqubits-w}$ Pauli measurements, 
we compute the expected squared-$\ell_2$ distance between the coefficients ${\sum_{P}|\hat{\alpha}_P-\alpha_P|^2}$,
which is exactly $\hsnorm{\hat{\rho}-\rho}^2$. 
Finally we use Cauchy-Schwarz $\hsnorm{\hat{\rho}-\rho}\ge \tracenorm{\hat{\rho}-\rho}/\sqrt{\dims}$ and Jensen's inequality to obtain the error in trace distance. 
Details can be found in \cref{sec:pauli-upper}. 
\section{Preliminaries}
\subsection{Quantum state and POVM}
We use the Dirac notation $\qbit{\psi}$ to denote a vector in $\C^{\dims}$. $\qadjoint{\psi}\eqdef(\qbit{\psi})^\dagger$ is a row vector. $\qdotprod{\psi}{\phi}$ is the Hibert-Schmidt inner product of $\qbit{\psi}$ and $\qbit{\phi}$. We denote the set of all $\dims\times\dims$ Hermitian matrices by $\Herm{\dims}$. A $\dims$-dimensional quantum system is described by a positive-semidefinite Hermitian matrix $\rho\in\Herm{\dims}$ with $\Tr[\rho]=1$. We assume $\dims=2^{\nqubits}$ where $\nqubits$ is the number of qubits in the system.

Measurements are formulated as \emph{positive operator-valued measure} (POVM). Let $\outset$ be an outcome set. Then a POVM $\POVM=\{M_x\}_{x\in \outset}$, where $M_x$ is p.s.d. and $\sum_{x\in \outset}M_x=\eye_\dims$. Let $X$ be the outcome of measuring $\rho$ with $\POVM$, then the probability observing $x\in\outset$ is given by the \emph{Born's rule},
\[
\probaOf{X=x}=\Tr[\rho M_x].
\]
\subsection{Hibert space over linear operators}

\paragraph{Hilbert space over complex matrices}
The space of complex matrices $\C^{\dims\times\dims}$ is a Hilbert space with inner product $\hdotprod{A}{B}\eqdef\Tr[A^\dagger B], A, B\in \C^{\dims\times \dims}$. For Hermitian matrices $A,B$, $\hdotprod{A}{B}=\hdotprod{B}{A}\in \R$. Thus the subspace of Hermitian matrices $\Herm{\dims}$ is a \textit{real} Hilbert space (i.e. the associated field is $\R$) with the same matrix inner product. 

Vectorization defines a homomorphism between $\C^{\dims\times\dims}$ and $\C^{\dims^2}$. 
On the canonical basis $\{\qbit{j}\}_{j=0}^{\dims-1}$, $\VecOp(\qoutprod{i}{j})\eqdef \qbit{j}\otimes \qbit{i}$. For convenience we denote $\vvec{A}\eqdef\VecOp(A)$. It is straightforward that $\hdotprod{A}{B}=\vvdotprod{A}{B}$. 

\paragraph{(Linear) superoperators} Let $\mathcal{N}:\C^{\dims\times \dims}\mapsto \C^{\dims\times \dims}$ be a linear operator over $\C^{\dims\times \dims}$, which we refer to as superoperators\footnote{This is to distinguish from a matrix $A\in\C^{\dims\times \dims}$, which can be viewed as an operator over $\C^\dims$. Indeed an operator over $\C^{\dims\times \dims}$ need not be linear, but we only deal with linear ones in this work, so we drop ``linear'' for brevity.}. Every superoperator $\mathcal{N}$ has a matrix representation $\Choi(\mathcal{N})\in \C^{\dims^2\times\dims^2}$ that satisfies $\vvec{\mathcal{N}(X)}=\Choi(\mathcal{N})\vvec{X}$ for all matrices $X\in\C^{\dims\times\dims}$. It can be verified that for the measurement information channel $\Luders_{\POVM}$ in~\cref{def:mic}, $\Choi_{\POVM}\vvec{A}=\vvec{\Luders_{\POVM}(A)}$.

\paragraph{Schatten norms} Let $\Lambda=(\lambda_1, \ldots, \lambda_\dims)\ge 0$ be the \emph{singular values} of a linear operator $A$, which can be a matrix or a superoperator. {For Hermitian matrices, the singular values are the absolute values of the eigenvalues.} Then for $p\ge 1$, the \emph{Schatten $p$-norm} is defined as 
$
\|A\|_{S_p}\eqdef \|\Lambda\|_p
$. The Schatten norms of a superoperator $\mathcal{N}$ and its matrix representation $\Choi(\mathcal{N})$ match exactly, $\|\mathcal{N}\|_{S_p}=\|\Choi(\mathcal{N})\|_{S_p}$. Some important examples are trace norm $\tracenorm{A}\eqdef\|A\|_{S_1}$, Hilbert-Schmidt norm $\hsnorm{A}\eqdef\|A\|_{S_2}=\sqrt{\hdotprod{A}{A}}$, and operator norm $\opnorm{A}\eqdef\|A\|_{S_\infty}=\max_{i=1}^\dims\lambda_i$.

\subsection{Problem setup}
There are $\ns$ copies of $\rho$ and a random seed $R$. We can apply adaptive measurements $\POVM^\ns = (\POVM_1, \ldots, \POVM_\ns)$ to each copy, where $\POVM_i=\{M_x^i\}_{x\in\cX}$. Let $\out_0=R$, and for $i\ge 1$ let $\out_i$ be the outcome of measuring the $i$th copy with $\POVM_i$. Define $\out^t=(\out_0,\out_1, \ldots, \out_t)$. Then we can write $\POVM_i=\POVM_i(\out^{i-1})$.

\paragraph{Tomography} The goal is to design a measurement scheme $\POVM^\ns$ and an estimator $\qest:\outset^\ns\mapsto \Herm{\dims}$ such that
\[
\inf_{\rho}\probaOf{\tracenorm{\qest(\out^\ns)-\rho}\le \eps}\ge0.99.
\]

Given measurements $\POVM^\ns$, when the state is $\rho$, the distribution of $\out_i,i\ge 1$ is determined by Born's rule,
\begin{equation}
    \p_\rho^{\out_i|\out^{i-1}}(x)=\Tr[M_x^i\rho],\label{equ:distr-i-cond}
\end{equation}
which depends on all previous outcomes and the random seed $R$. For $1\le t\le \ns$, we further define $\p_\rho^{\out^t}$ as the distribution of $\out^t$ when the state is $\rho$.

In practice, we may have restrictions on the types of measurements that can be applied. We use $\povmset$ to denote the set of allowable measurements for each copy. Define the following quantities which are related to the norms of the measurement information channel in this family of measurements,
\begin{align}
\norm{\povmset}=\sup_{\POVM\in\povmset}\norm{\Luders_{\POVM}},
\label{equ:max-povm-norm}
\end{align}
where $\norm{\cdot}$ can be any norms for linear operators, including $\tracenorm{\cdot}, \hsnorm{\cdot}$, and $\opnorm{\cdot}$.

\subsection{Pauli measurements}
For a single-qubit system, the Pauli operators $\Sigma=\{\pauliX, \pauliY, \pauliZ\}$ are defined in \cref{equ:pauli-ops} An important property is that for $P,Q\in \Sigma\cup \{\eye_2\}$,
\begin{equation}
P^2=\eye_2,\;\Tr[PQ]=2\indic{P=Q}.    
\label{equ:pauli-property}
\end{equation}

 Let $\qbit{0}$ and $\qbit{1}$ be the computation basis for a single-qubit system. Define the following states
\[
\qbit{+}\eqdef \frac{1}{\sqrt{2}}(\qbit{0}+\qbit{1}),\; \qbit{-}\eqdef \frac{1}{\sqrt{2}}(\qbit{0}-\qbit{1}),
\;
\qbit{+\img}\eqdef \frac{1}{\sqrt{2}}(\qbit{0}+\img\qbit{1}),\; \qbit{-\img}\eqdef \frac{1}{\sqrt{2}}(\qbit{0}-\img\qbit{1}).
\]
Note that both $\{\qbit{+},\qbit{-}\}$ and $\{\qbit{+\img},\qbit{-\img}\}$ are orthonormal bases for $\C^2$. Furthermore,

\[
\sigma_X\qbit{+}=\qbit{+}, \quad\sigma_X\qbit{-}=-\qbit{-},
\]
\[
\sigma_Y\qbit{+\img}=-\qbit{+\img}, \quad\sigma_X\qbit{-\img}=-\qbit{-\img},
\]
\[
\sigma_Z\qbit{0}=\qbit{0}, \quad\sigma_Z\qbit{1}=-\qbit{1},
\]
Thus Pauli operators have eigenvalues of either $1$ or $-1$. We refer to $\{\qbit{+},\qbit{-}\}$ as the $X$ basis, $\{\qbit{+\img},\qbit{-\img}\}$ as the $Y$ basis, and $\{\qbit{0},\qbit{1}\}$ as the $Z$ basis because they are the eigenvectors of the respective Pauli operators.

\paragraph{Pauli (basis) measurement} For an $\nqubits$-qubit system, we can independently measure in the $X$, $Y$, or $Z$ basis for each qubit. This results in a basis measurement with $2^\nqubits$ outcomes, which we denote by $\{-1, 1\}^{\nqubits}$. We call this a \emph{Pauli basis measurement}, or Pauli measurement in short. Formally, given $P=\sigma_1\otimes\cdots\otimes\sigma_{\nqubits}$ where $\sigma_i\in \Sigma=\{\pauliX, \pauliY, \pauliZ\}$, the Pauli measurement is defined as
\begin{equation}
    \POVM_{P}=\{M_x^P\}_{x\in \{-1, 1\}^{\nqubits}},\; M_x^\sigma\eqdef \bigotimes_{j=1}^{\nqubits}\frac{\eye_2+x_j\sigma_j}{2},x=(x_1, \ldots, x_{\nqubits})\in\{-1, 1\}^{\nqubits}.
    \label{equ:pauli-measurement}
\end{equation}

\paragraph{Pauli observable} A weaker type of Pauli measurement is defined by the Pauli observables.
\begin{definition}
    A \emph{Pauli observable} $P\in \C^{\dims\times\dims}$ can be written as
\[
P=\sigma_1\otimes\cdots\otimes \sigma_{\nqubits}, \sigma_j\in\Sigma\cup\{\eye_2\}, P\ne \eye_\dims.
\]
The \emph{weight} of $P$ is the number of $\sigma_j$'s in $P$ such that $\sigma_j\ne \eye_2$, denoted as $w(P)$. We sometimes omit $P$ when the Pauli observable is clear from context.
\end{definition}

The set of Pauli observables $\pauliObsSet\eqdef (\Sigma\cup \eye_2)^{\otimes\nqubits}\setminus\{\eye_\dims\}$ consists of $4^{\nqubits}-1=\dims^2-1$ matrices. Each $P\in\pauliObsSet$ defines a 2-outcome POVM,
\[
\ObsPOVM_P=\{M_{-1}^P, M_{1}^P\},\; M_{-1}^P=\frac{\eye_\dims - P}{2}, M_1^P=\frac{\eye_\dims+P}{2}.
\]

We have the standard fact about Pauli observables,
\begin{fact}
\label{fact:pauli}
    Let $P, Q\in \mathcal{P}$ be two Pauli observables. Then, 
\[
P^2=\eye_\dims, \;\Tr[P]=0,\; \Tr[PQ]=\hdotprod{P}{Q}=\dims\indic{P=Q}.
\]
\end{fact}
Therefore, the set $\pauliObsSet\cup \{\eye_\dims\}$ forms an orthogonal basis for $\Herm{\dims}$. We can represent any state $\rho$ as,
\begin{equation}
    \rho = \frac{\eye_\dims}{\dims}+\sum_{P\in \pauliObsSet}\frac{\Tr[\rho P]P}{\dims}.
    \label{equ:state-pauli}
\end{equation}
\subsection{Probability distances}
Let $\p$ and $\q$ be distributions over a finite domain $\mathcal{X}$. The \emph{total variation distance} is defined as 
\[
\totalvardist{\p}{\q}\eqdef\sup_{S\subseteq\mathcal{X}}(\p(S)-\q(S))=\frac{1}{2}\sum_{x\in\mathcal{X}}|\p(x)-\q(x)|.
\]
The KL-divergence  is
\[
\kldiv{\p}{\q}\eqdef\sum_{x\in\mathcal{X}}\p(x)\log\frac{\p(x)}{\q(x)}.
\]
The symmetric KL-divergence is 
$\kldivsym{\p}{\q}=\frac{1}{2}(\kldiv{\p}{\q}+\kldiv{\q}{\p})$. The chi-square divergence
\[
\chisquare{\p}{\q}\eqdef \sum_{x\in\mathcal{X}}\frac{(\p(x)-\q(x))^2}{\q(x)}.
\]
By Pinsker's inequality and concavity of logarithm,
\[
2\totalvardist{\p}{\q}^2\le \kldiv{\p}{\q}\le \chisquare{\p}{\q}.
\]
We define $\ell_p$ distance as $
\norm{\p-\q}_p\eqdef\Paren{\sum_{x\in\mathcal{X}}{|\p(x)-\q(x)|^p}}^{1/p}.
$

\section{Adaptive tomography lower bound}
\subsection{Lower bound construction}
\begin{definition}
 \label{def:perturbation}
     Let $\dims^2/2\le\ell\le\dims^2-1$ and $\hbasis=(V_1, \ldots, V_{\dims^2}=\frac{\eye_\dims}{\sqrt{\dims}})$ be an orthonormal basis of $\Herm{\dims}$, and $\cd$ be a universal constant. Let  $\ptb=(\ptb_1, \ldots, \ptb_\ell)$ be uniformly drawn from $\{-1, 1\}^\ell$,
     \begin{equation}
         \Delta_{\ptb} = \frac{\cd\eps}{\sqrt{\dims}}\cdot\frac{1}{\sqrt{\ell}}\sum_{i=1}^\ell \ptb_iV_i, \quad \barDelta_{\ptb}= \Delta_{\ptb}\min\left\{1, \frac{1}{2\dims \opnorm{\Delta_{\ptb}}}\right\},
         \label{equ:delta_z}
     \end{equation}
     Finally we set $\sigma_{\ptb}=\qmm + \barDelta_{\ptb}$ whose distribution we denote as $\ptbDistr(\hbasis)$.
 \end{definition}

The construction adds independent binary perturbations to $\qmm$ along $\ell$ orthogonal trace-0 directions. With appropriate constant $\cd$, $\ptbDistr(\hbasis)$ has an exponentially small probability mass outside the set $\mathcal{P}_\eps\eqdef\{\rho: \tracenorm{\rho-\qmm}>\eps\}$.
\begin{theorem}[{\cite[Corollary 4.4]{liu2024role}}]
\label{prop:perturbation-trace-distance}
    Let $\cd= 10\sqrt{2}$, $\ell\ge \dims^2/2$, $\eps<1/200$. Then for $\sigma\sim \ptbDistr(\hbasis)$,  $\|\sigma-\qmm\|_1\ge \eps$ with probability at least $1-2\exp(-\dims)$. 
\end{theorem}

This is the result of a random matrix concentration.
\begin{restatable}[{\cite[Theorem 4.2]{liu2024role}}]{theorem}{randmatopnorm}
\label{thm:rand-mat-opnorm-concentration}
    Let $V_1, \ldots, V_{\dims^2}\in\C^{\dims\times \dims}$ be an orthonormal basis of $\C^{\dims\times \dims}$ and $\ptb_1, \ldots, \ptb_{\dims^2}\in\{-1, 1\}$ be independent symmetric Bernoulli random variables. Let $W=\sum_{i=1}^{\ell}\ptb_iV_i$ where $\ell\le \dims^2$. For all $\alpha>0$, there exists $\cop_\alpha$, {which is increasing in $\alpha$} such that
    \[
    \probaOf{\opnorm{W}>\cop_\alpha\sqrt{\dims}}\le 2\exp\{-\alpha\dims\}.
    \]
\end{restatable}

Let $z\sim\{-1,1\}^{\ell}$ and $\sigma_z\sim \ptbDistr(\hbasis)$ be defined in \cref{def:perturbation}. Use the shorthand $\p_z^{\out_i|\out^{i-1}}=\p_{\sigma_z}^{\out_i|\out^{i-1}}$. 
We define the following mixtures,
\begin{equation}
    \p_{+i}^{\out^\ns}\eqdef \frac{1}{2^{\ell-1}}\sum_{\ptb:\ptb_i=+1}\p_z^{\out^\ns},\quad  \p_{-i}^{\out^\ns}\eqdef \frac{1}{2^{\ell-1}}\sum_{\ptb:\ptb_i=-1}\p_z^{\out^\ns}.
\end{equation}
Which are the distributions of outcomes $\out^\ns$ when we fix the $i$th coordinate to be $+1$ and $-1$ respectively. Then we can define,
\begin{equation}
\label{equ:out-distr}
\q^{\out^\ns}\eqdef\frac{1}{2^\ell}\sum_{z\in\{-1,1\}^\ell}\p_z^{\out^\ns}=\frac{1}{2}(\p_{+i}^{\out^\ns}+\p_{-i}^{\out^\ns}).  
\end{equation}

This is exactly the distribution of $\out^\ns$ when $\sigma_z\sim \ptbDistr(\hbasis)$ and outcomes $\out^\ns$ are obtained by measuring $\sigma_z^{\otimes\ns}$ with the adaptive scheme $\POVM^\ns$.

\subsection{Mutual information upper bound via MIC}
The following theorem bounds the mutual information in terms of the measurement information channel.
\begin{theorem}
\label{thm:avg-MI-upper}
    Let $\sigma_\ptb\sim\ptbDistr(\hbasis)$ where $\ptb\sim\{-1,1\}^{\ell}$, $\out^\ns$ be the outcomes after applying $\POVM^\ns$. Then for $\dims\ge 1024$ and all $t\in[\ns]$,
    \begin{align}
         \frac{1}{\ell}\sum_{i=1}^{\ell}\mutualinfo{\ptb_i}{\out^t}&\le \frac{8 tc^2 \eps^2}{\ell^2}  \sup_{\POVM\in {\povmset}}{\sum_{i=1}^\ell \vadj{V_i} \Choi_{\POVM} \vvec{V_i}} +16\exp\{-\alpha\dims\}tc^2\eps^2 \label{equ:avg-MI-partial}\\
         &\le \frac{16tc^2\eps^2}{\ell^2}\tracenorm{\povmset}.\label{equ:avg-MI-tracenorm}
    \end{align}
\end{theorem}

\begin{proof}
    We start with the fundamental fact that mutual information is the conditional KL-divergence between the conditional distribution given the marginal $x^t$: $\p^{\out^t}_{z_i}$ for $1 \leq i \leq n$ and the marginal distribution $\q^{\out^t}$,
\begin{align*}
    I(z_i;x^t) &= \kldiv{\p^{\out^t}_{z_i}}{\q^{\out^t} \mid z_i} = \expectDistrOf{z_i}{\kldiv{\p^{\out^t}_{z_i}}{\q^{\out^t}}} \\ &= \frac{1}{2} \kldiv{\p_{+i}^{\out^t}}{\q^{\out^t}} + \frac{1}{2} \kldiv{\p_{-i}^{\out^t}}{\q^{\out^t}} \\ 
    &= \frac{1}{2} \kldiv{\p_{+i}^{\out^t}}{\frac{\p_{+i}^{\out^t} + \p_{-i}^{\out^t}}{2}} + \frac{1}{2} \kldiv{\p_{-i}^{\out^t}}{\frac{\p_{+i}^{\out^t} + \p_{-i}^{\out^t}}{2}}.
\end{align*}
Thus, by convexity, 
\begin{align}
    I(z_i;x^t) &\leq \frac{1}{4} \left[\kldiv{\p^{\out^t}_{+i}}{\p^{\out^t}_{+i}}+\kldiv{\p^{\out^t}_{-i}}{\p^{\out^t}_{+i}} \right] = \frac{1}{2} \kldivsym{\p^{\out^t}_{+i}}{\p^{\out^t}_{-i}} \label{pf:MI-bound}.
\end{align}
Where the last inequality comes from the convexity of KL-divergence with respect to its second argument. Given this symmetric KL-divergence between the mixture distribution conditioned on the i-th perturbation, we can further narrow the correlation between the measurement outcomes and the perturbation with the change in measurement outcome distribution when $z_i$ is flipped. We apply chain rule on the symmetric KL-divergence to allow us to isolate the per measurement round divergence,
\begin{align}
    \kldivsym{\p^{\out^t}_{+i}}{\p^{\out^t}_{-i}} = \sum_{j=1}^{t} \expectDistrOf{\q^{x^\ns}}{\kldivsym{\p^{\out_{j}|\out^{j-1}}_{+i}}{\p^{\out_{j}|\out^{j-1}}_{-i}}}.
    \label{eq:per-round-divergence}
\end{align}
We bound the symmetric KL by the chi-squared divergence,
\begin{align*}
    \kldivsym{\p^{\out_{j}|\out^{j-1}}_{+i}}{\p^{\out_{j}|\out^{j-1}}_{-i}} &\le \chisquare{\p^{\out_{j}|\out^{j-1}}_{+i}}{\p^{\out_{j}|\out^{j-1}}_{-i}} \\
    &\leq \frac{1}{2^{l-1}} \sum_{z \in \{+1,-1\}^\ell} \; \chisquare{\p^{\out_{j}|\out^{j-1}}_{z}}{\p^{\out_{j}|\out^{j-1}}_{z^{\oplus i}}}  \\
    &= \frac{1}{2^{\ell-1}} \sum_{z \in \{+1,-1\}^\ell} \; \expectDistrOf{X \sim \p^{\out_{j}|\out^{j-1}}_{z^{\oplus i}}}{\delta_j(X)^2}.
\end{align*}
 Where the last inequality is from the joint convexity of f-divergences. $\delta_t(X)$ follows the definition,
\begin{align*}
    \delta_j(x) \eqdef \frac{\p^{\out_{j}|\out^{j-1}}_{z}(x)-\p^{\out_{j}|\out^{j-1}}_{z^{\oplus i}}(x)}{\p^{\out_{j}|\out^{j-1}}_{z^{\oplus i}}(x)}.
\end{align*}
Furthermore, $\delta_j$ term can be bounded by extracting the MIC channel,
\begin{align*}
    \delta_j(x) &= \frac{\Tr[M_x^j (\qmm + \barDelta_{\ptb})] - \Tr[M_x^j (\qmm + \barDelta_{\ptb^{\oplus i}})]}{\Tr[M_x^j (\qmm + \barDelta_{\ptb^{\oplus i}})]} \\
    &= \frac{\Tr[M_x^j (\barDelta_{\ptb} - \barDelta_{\ptb ^{\oplus i}})]}{\Tr[M_x^j (\qmm + \barDelta_{\ptb^{\oplus i}})]}.
\end{align*}
Therefore, we plug $\delta_j(X)$ into the expectation and noting that $\p^{\out_{j}|\out^{j-1}}_{z^{\oplus i}}(x) = \Tr[M_x^j (\qmm + \barDelta_{\ptb^{\oplus i}})]$,
\begin{align*}
    \expectDistrOf{X \sim \p^{\out_{j}|\out^{j-1}}_{z^{\oplus i}}}{\delta_j(X)^2} &= \sum_{x \in X} \frac{\Tr[M_x^j (\barDelta_{\ptb} - \barDelta_{\ptb ^{\oplus i}})]^2}{\Tr[M_x^j (\qmm + \barDelta_{\ptb^{\oplus i}})]} \\
    &= \sum_{x \in X} \frac{\Tr[(\barDelta_{\ptb} - \barDelta_{\ptb ^{\oplus i}}) M_x^j] \Tr[(\barDelta_{\ptb} - \barDelta_{\ptb ^{\oplus i}}) M_x]^{\dagger}}{\Tr[M_x^j (\qmm + \barDelta_{\ptb^{\oplus i}})]} \\
    &= \sum_{x \in X} \frac{\Tr[(\barDelta_{\ptb} - \barDelta_{\ptb ^{\oplus i}}) M_x^j] \Tr[(\barDelta_{\ptb} - \barDelta_{\ptb ^{\oplus i}}) M_x^j]^{\dagger}}{\Tr[M_x^j (\qmm + \barDelta_{\ptb^{\oplus i}})]} \\
    &= \sum_{\out \in \mathcal{X}} \frac{\vvdotprod{(\barDelta_{\ptb} - \barDelta_{\ptb ^{\oplus i}})}{M_x^j}\vvdotprod{M_x^j}{(\barDelta_{\ptb} - \barDelta_{\ptb ^{\oplus i}})}}{\Tr[M_x^j(\qmm + \barDelta_{\ptb^{\oplus i}})]}.
\end{align*}
Note that $\qmm + \barDelta_{\ptb^{\oplus i}} \succcurlyeq \frac{1}{2} \qmm \implies \Tr[M_x^j(\qmm + \barDelta_{\ptb^{\oplus i}})] \geq \Tr[M_x^j(\frac{1}{2} \qmm)] = \frac{1}{2 \dims}\Tr[M_x^j]$. This statement comes from the fact that $\opnorm{\barDelta_{\ptb^{\oplus i}}} \leq \frac{1}{2 \dims}$~\eqref{equ:delta_z},
\begin{align*}
   \expectDistrOf{X \sim \p^{\out_{j}|\out^{j-1}}_{z^{\oplus i}}}{\delta_j(X)^2} &\le  2 \dims \sum_{\out \in \mathcal{X}} \frac{\vvdotprod{(\barDelta_{\ptb} - \barDelta_{\ptb ^{\oplus i}})}{M_x^j}\vvdotprod{M_x^j}{(\barDelta_{\ptb} - \barDelta_{\ptb ^{\oplus i}})}}{\Tr[M_x^j]} \\
   &=  2 \dims \vadj{(\barDelta_{\ptb} - \barDelta_{\ptb ^{\oplus i}})} \sum_{\out \in \mathcal{X}} \frac{\vvec{M_x^j}\vadj{M_x^j}}{\Tr[M_x^j]} \vvec{(\barDelta_{\ptb} - \barDelta_{\ptb ^{\oplus i}})}.
\end{align*}
We can then apply this bound to the per-round symmmetric KL-divergence,
\begin{align*}
   \kldivsym{\p^{\out_{j}|\out^{j-1}}_{+i}}{\p^{\out_{j}|\out^{j-1}}_{-i}}&\le \frac{1}{2^{\ell-1}} \sum_{z \in \{+1,-1\}^l} \; 2 \dims \vadj{(\barDelta_{\ptb} - \barDelta_{\ptb ^{\oplus i}})} \Sigma_{\out \in \mathcal{X}} \frac{\vvec{M_x^j}\vadj{M_x^j}}{\Tr[M_x^j]} \vvec{(\barDelta_{\ptb} - \barDelta_{\ptb ^{\oplus i}})} \\
   &= 4 \dims \expectDistrOf{z\sim\{-1, 1\}^{\ell}}{\vadj{(\barDelta_{z} - \barDelta_{z ^{\oplus i}})} \Choi_{\POVM_j} \vvec{(\barDelta_{\ptb} - \barDelta_{z ^{\oplus i}})}},
\end{align*}
where $z$ is drawn uniformly from $\{-1,1\}^{\ell}$.
Another key to this bound is that we have a concentration on the operator norm of the perturbation matrix such that the operator norm lies in the boundary (within some constant)  with exponentially decreasing probability, see \cref{thm:rand-mat-opnorm-concentration}. 
Intuitively, this means that it is rare that all of the $z_i$ components are selected in a way where eigenvectors of the $z_i V_i$ components are aligned, thus resulting in an equal contribution to the total perturbation from each $z_i V_i$ component. 
As a result, this concentration perspective allows us to see that flipping a single $z_i V_i$ entry will dictate a perturbation outcome with high probability.
For convenience, we define the concentration set for the perturbation parameters,
\begin{align}
    \mathcal{G} := \{z \in \{1,1\}^{\ell} \; | \; \opnorm{W_z} \leq \kappa_\alpha \sqrt{d}\} \label{eq:concentrated-set},   
\end{align}
where $\alpha$ is a positive constant and $\kappa_\alpha$ is a positive constant non-decreasing in $\alpha$. By \cref{thm:rand-mat-opnorm-concentration},
\[
\probaOf{z \in \mathcal{G}} \geq 1 - 2\exp\{- \alpha d\}.
\]
For more details on the constants involved, see Lemma 21 of \cite{liu2024role}. 
We then condition between the possible cases of perturbations with law of iterative expectation,

\begin{align}
    \kldivsym{\p^{\out_{j}|\out^{j-1}}_{+i}}{\p^{\out_{j}|\out^{j-1}}_{-i}} &\le 4 \dims \expect{\expectDistrOf{z}{\vadj{(\barDelta_{\ptb} - \barDelta_{\ptb ^{\oplus i}})} \Choi_{\POVM_j} \vvec{(\barDelta_{\ptb} - \barDelta_{\ptb ^{\oplus i}})} \;| \indic{z \in \mathcal{G}}}}\label{pf:KL-div-concentration}.
\end{align}
  Now, it suffices to bound the peturbation for when $z \in G$ and $z \notin G$. When $z \in G$, the following bound holds for $\eps \leq \frac{1}{4(\kappa_\alpha+1)}$,
\begin{align}
\opnorm{\Delta_z} &= \frac{c\eps}{\sqrt{\dims \ell}} \opnorm{W_z} \leq  \frac{\kappa_\alpha c\eps}{\sqrt{\ell}} \leq \frac{2 \kappa_\alpha c\eps}{\dims} \leq \frac{1}{2 \dims} .\label{eq:op-norm-bound}
\end{align}
In addition, the following holds when the i-th bit is flipped,
\begin{align*}
        \opnorm{W_{z ^{\oplus i}}} &= \opnorm{W_z - 2 z_i V_i} \leq \opnorm{W_z} + \opnorm{-2 z_i V_i} \\
    &\le \kappa_\alpha \sqrt{\dims} + 2 \leq (\kappa_\alpha + 1) \sqrt{d} \\
    \implies \opnorm{\Delta_{z^{\oplus i}}} &= \frac{c\eps}{\sqrt{\dims \ell}} \opnorm{W_{z^{\oplus i}}} \leq  \frac{(\kappa_\alpha + 1) c\eps}{\sqrt{\ell}} \leq \frac{2(\kappa_\alpha + 1) c\eps}{\dims} \leq \frac{1}{2\dims}.
\end{align*}
The second inequality follows because $ \opnorm{V_i}^2 = \|V_i\|_{S_\infty}^2 \leq\|V_i\|_{S_2}^2 = \hdotprod{V_i}{V_i} = 1$. For $z \in \mathcal{G}$, we have that $\opnorm{\Delta_{z}}, \opnorm{\Delta_{z^{\oplus i}}} \leq \frac{1}{2 \dims}$. This results in $\barDelta_{z ^{\oplus i}} = \Delta_{z ^{\oplus i}}, \;\barDelta_{z} = \Delta_{z}$, by definition of the normalization factor in \cref{equ:delta_z}. Thus,
\begin{align*}
 \vadj{(\barDelta_{z} - \barDelta_{z ^{\oplus i}})} \Choi_{\POVM_j} \vvec{(\barDelta_{z} - \barDelta_{z ^{\oplus i}})} &=  \vadj{(\Delta_{z} - \Delta_{z ^{\oplus i}})}
 \Choi_{\POVM_j} \vvec{(\Delta_{z} - \Delta_{z ^{\oplus i}})} \\
 &=\vadj{\frac{c\eps}{\sqrt{\dims \ell}} 2 z_i V_i}
 \Choi_{\POVM_j}
 \vvec{\frac{c\eps}{\sqrt{\dims \ell}} 2 z_i V_i} = \frac{4 c^2 \eps^2 z_{i}^2}{\dims \ell}  = \frac{4 c^2 \eps^2}{\dims \ell} \vadj{V_i} \Choi_{\POVM_j} \vvec{V_i}.
\end{align*}
We will later see that this will result in the trace decomposition of $\Choi_{\POVM_j}$ under the vectorized version of the orthornormal Hilbert basis $\hbasis$. Now, we will apply a more crude bound for the low-concentration set $z \notin \mathcal{G}$. We start by bounding the Hilbert-Schmidt norm of the perturbation matrix for every $z \in \{-1,1\}^\ell$,
\begin{align*}
\hsnorm{\barDelta_{z}} &= \sqrt{\vvdotprod{\barDelta_{z}}{\barDelta_{z}}} \\
&= \sqrt{\frac{c^2\eps^2}{\dims \ell}\vvdotprod{\sum_{i=1}^\ell \min\left\{1, \frac{1}{2\dims \opnorm{\Delta_{\ptb}}}\right\} z_i V_i}{\sum_{i=1}^\ell \min\left\{1, \frac{1}{2\dims \opnorm{\Delta_{\ptb}}}\right\} z_i V_i}}\\
&= \sqrt{\frac{c^2 \eps^2}{\dims \ell} \sum_{i \neq j} \min\left\{1, \frac{1}{2\dims \opnorm{\Delta_{\ptb}}}\right\}^2 z_i z_j \vvdotprod{V_i}{V_j} + \sum_{i}^\ell \min\left\{1, \frac{1}{2\dims \opnorm{\Delta_{\ptb}}}\right\}^2 z_i^2 \vvdotprod{V_i}{V_i}} \\
&= \sqrt{\frac{c\eps^2}{\dims \ell} \sum_{i}^\ell \min\left\{1, \frac{1}{2\dims \opnorm{\Delta_{\ptb}}}\right\}^2} \leq  \sqrt{\frac{c\eps^2}{\dims \ell} \sum_{i}^\ell 1} = \frac{c \eps}{\sqrt{\dims}}.
\end{align*} 
Where last line holds from the orthonormality of the perturbation basis. Now, we can use triangle inequality of the Hilbert-Schmidt norm to get the bound on the Hilbert-Schmidt norm of the difference between the perturbation matrices.
\begin{align*}
    \vadj{(\barDelta_{z} - \barDelta_{z ^{\oplus i}})} \Choi_{\POVM_j} \vvec{(\barDelta_{z} - \barDelta_{z ^{\oplus i}})} 
    &\le \opnorm{\Choi_{\POVM_j}}\hsnorm{\barDelta_{z} - \barDelta_{z ^{\oplus i}}}^2  \\
    &\leq 2(\hsnorm{\barDelta_{z}}^2 + \hsnorm{\barDelta_{z ^{\oplus i}}}^2) \\
    &\leq \frac{4c^2 \eps^2}{\dims}.
\end{align*}
The first step is due to the definition of operator norm. The second step is because $\opnorm{\Choi_{\POVM_j}} \le 1$, triangle inequality, and $(a+b)^2\le 2(a^2+b^2)$. We can further bound the symmetric KL-divergence in \cref{pf:KL-div-concentration},
 \begin{align*}
     \kldivsym{\p^{\out_{j}|\out^{j-1}}_{+i}}{\p^{\out_{j}|\out^{j-1}}_{-i}} &\le 4 \dims \left[\probaOf{z \in \mathcal{G}} \frac{4 c^2 \eps^2}{\dims\ell} \vadj{V_i} \Choi_{\POVM_j} \vvec{V_i} + (1 - \probaOf{z \in \mathcal{G}})   \frac{4 c^2 \eps^2}{\dims}\right] \\
     &= \probaOf{z \in \mathcal{G}} \frac{16 c^2 \eps^2}{\ell} \vadj{V_i} \Choi_{\POVM_j} \vvec{V_i} + (1-\probaOf{z \in \mathcal{G}}) 16 c^2 \eps^2.
 \end{align*} 

Thus combining with \eqref{pf:MI-bound} \eqref{eq:per-round-divergence},
\begin{align*}
    \frac{1}{\ell}\sum_{i=1}^{\ell}\mutualinfo{\ptb_i}{\out^t} &\le \frac{1}{2\ell}\sum_{i=1}^{\ell}\sum_{j=1}^{t} \expectDistrOf{\q^{\out^\ns}}{\kldivsym{\p^{\out_{j}|\out^{j-1}}_{+i}}{\p^{\out_{j}|\out^{j-1}}_{-i}}} \\
    &\leq \probaOf{z \in \mathcal{G}} \frac{8 c^2 \eps^2}{\ell^2} \sum_{j=1}^{t} \sum_{i=1}^{\ell} \expectDistrOf{\q^{\out^\ns}}{ \vadj{V_i} \Choi_{\POVM_j} \vvec{V_i}}
    + (1 - \probaOf{z \in \mathcal{G}}) \sum_{j=1}^{t} \sum_{i=1}^{\ell}  \frac{8c^2\eps^2}{\ell} \\
    &\le  \frac{8 tc^2 \eps^2}{\ell^2}  \expectDistrOf{\q^{\out^\ns}}{\frac{1}{t}\sum_{j=1}^{t}\sum_{i=1}^\ell \vadj{V_i} {\Choi}_{\POVM_j} \vvec{V_i}} +16\exp\{-\alpha\dims\}tc^2\eps^2\\
    &\le \frac{8 tc^2 \eps^2}{\ell^2}  \sup_{\POVM\in\povmset}\sum_{i=1}^\ell \vadj{V_i} {\Choi}_{\POVM} \vvec{V_i} +16\exp\{-\alpha\dims\}tc^2\eps^2,
\end{align*}
The second term in the final step is due to \cref{thm:rand-mat-opnorm-concentration}. This proves \eqref{equ:avg-MI-partial} in \cref{thm:avg-MI-upper}.

We continue to derive the remaining expression \eqref{equ:avg-MI-tracenorm}. We use the fact that for any matrix $A\in\C^{\dims\times\dims}$ and an orthonormal basis $\qbit{u_1}, \ldots, \qbit{u_\dims}$,
\[
\Tr[A]=\sum_{i=1}^{\dims}\matdotprod{u_i}{A}{u_i}.
\]
Combining with the fact that ${\Choi}_{\POVM}$ is p.s.d., we have
\[
\sum_{i=1}^{\ell} \vadj{V_i} {\Choi}_{\POVM} \vvec{V_i}\le \sum_{i=1}^{\dims^2}\vadj{V_i}\Choi_{\POVM} \vvec{V_i}=\Tr[\Choi_{\POVM}]=\tracenorm{\Choi_{\POVM}}.
\]
Therefore, continuing from \eqref{equ:avg-MI-partial},

\begin{align*}
     \frac{1}{\ell}\sum_{i=1}^{\ell}\mutualinfo{\ptb_i}{\out^t} 
    &\le \frac{8 t c^2 \eps^2}{\ell^2} \tracenorm{\povmset} + 16\exp\{-\alpha\dims\}tc^2\eps^2\\
    & \le \frac{16 t c^2 \eps^2}{\ell^2} \tracenorm{\povmset}.
\end{align*}
The first step is from the definition of $\tracenorm{\povmset}$ in \eqref{equ:max-povm-norm}. The second step holds as $\tracenorm{\povmset}\ge 1$ and $\exp\{-\alpha d\} \leq \frac{1}{d^4}$ when $d \geq 1024$. 
\end{proof}

\subsection{Mutual information lower bound}
We state some useful bounds on mutual information.
\begin{lemma}[{\cite[Lemma 10]{ACLST22iiuic}}]
\label{lem:MI-lower}
    Let $Z\in\{-1, 1\}^\ab$ be drawn uniformly and $Z-Y-\hat{Z}$ be a Markov chain where $\hat{Z}$ is an estimate of $Z$. Let $h(t)\eqdef -t\log t-(1-t)\log(1-t)$, then for each $i\in[\ab]$,
    \[
    \mutualinfo{Z_i}{Y}\ge 1-h(\probaOf{Z_i\ne \hat{Z}_i}).
    \]
\end{lemma}

The following lemma is an Assouad-type lower bound on the average mutual information. 
\begin{lemma}
\label{lem:avg-MI-lower}
    Let $\sigma_\ptb\sim\ptbDistr(\hbasis)$ where $\ptb\sim\{-1,1\}^{\ell}$, $\out^\ns$ be the outcomes after applying $\POVM^\ns$ to $\sigma_\ptb^{\otimes\ns}$, and $\qest$ be an estimator using $\out^\ns$ that achieves an accuracy of $\eps$. Then,
    \begin{equation}
        \frac{1}{\ell}\sum_{i=1}^{\ell}\mutualinfo{Z_i}{Y}\ge\frac{1}{100}.
    \end{equation}
\end{lemma}

Combining \cref{lem:avg-MI-lower} and \cref{thm:avg-MI-upper} proves the interactive lower bound for tomography.
\begin{proof}
   The idea behind this bound is that any good estimation $\qest$ of the parameterized state $\sigma_\ptb$ is close in the sense that the closest parameterized $\sigma_{\zest}$ to $\qest$ should also be sufficiently close. Then, we can relate the distance $\tracenorm{\sigma_{\ptb} - \sigma_{\zest}}$ to the hamming distance in $\sum_{i=1}^\ell \indic{z_i \neq \zest_i}$. Once this relation is established, then a optimal tomography algorithm should also have low probability of error in estimating $z$ with $\zest$. Thus, leading to lower bound of mutual information with the application of \cref{lem:MI-lower}. We begin by first bounding the error between the "parameterized version" of the estimator and $\sigma_{\hat{\ptb}}$,
   \begin{align*}
        \zest &:= \argmin_{\ptb \in \{-1,1\}^\ell } \tracenorm{\sigma_{\ptb} - \qest}\\
        \tracenorm{\sigma_{\zest} - \sigma_{\ptb}} &\leq \tracenorm{\sigma_{\ptb} - \qest} + \tracenorm{\qest-\sigma_{\zest}} \leq 2 \tracenorm{\sigma_{\ptb} - \qest},
   \end{align*}
   where the last line holds since $\tracenorm{\qest-\sigma_{\zest}} \leq \tracenorm{\qest-\sigma_{\ptb}}$ by definition of $\hat{\ptb}$. Notice $ \tracenorm{\qest-\sigma_{\ptb}} \leq \eps \implies \tracenorm{\sigma_{\zest} - \sigma_{\ptb}} \leq 2\eps $. Thus, 
   $$\Pr[\tracenorm{\sigma_{\zest} - \sigma_{\ptb}} \leq 2 \eps] \ge \Pr[\tracenorm{\sigma_{\zest} - \sigma_{\ptb}} \leq \eps] \geq 0.99.$$
   Now, we will introduce a lemma that will allow us to construct a informaton-theoretic packing around this estimator. This is done by relating the trace distance and the hamming distance between Z parameters. We present the formal version of~\cref{lemma:hamm-separation-informal}

   \begin{lemma}[Trace distance Hamming separation] \label{lemma:hamm-packing}
       Consider $z \in \mathcal{G}$, where $\mathcal{G}$ is defined from \cref{eq:concentrated-set}. For any  $\hat{z} \in \left\{-1,1\right\}^{\ell}$,
       \begin{equation}
           \tracenorm{\sigma_\ptb - \sigma_{\zest}} \geq \frac{c \eps}{2\kappa_\alpha \ell} \ham{\ptb}{\zest}.
       \end{equation}
   \end{lemma}

   \begin{proof}
       
   Let $C_{z} := \min\left\{1, \frac{1}{2\dims \opnorm{\Delta_{z}}}\right\}$ and define the matrices,
   \[\Delta_{w} := \frac{c \eps}{\sqrt{d\ell}}  \sum_{i=1}^\ell \indic{z_i \neq \hat{z}_i} z_i V_i, \;\Delta_{c} := \frac{c \eps}{\sqrt{d\ell}} \sum_{i=1}^\ell \indic{z_i = \hat{z}_i} z_i V_i.
   \]
   Notice the trace norm of distance between perturbation matrices has the following form,
   \begin{align*}   
   \tracenorm{\sigma_{\zest} - \sigma_{\ptb}} & = \tracenorm{\barDelta_{\hat{\ptb}} - \barDelta_{\ptb}} \\
   &=\tracenorm{C_{\hat{z}} \Delta_{\hat{\ptb}} - C_{z} \Delta_{\ptb}} \\
   &= \frac{c \eps}{\sqrt{d\ell}} \tracenorm{(-C_z-C_{\zest}) \sum_{i=1}^\ell \indic{z_i \neq \zest_i} z_i V_i + (C_{\zest} - C_z) \sum_{i=1}^\ell \indic{z_i = \zest_i} z_i V_i} \\
   &= \tracenorm{(C_z+C_{\zest}) \Delta_w + (C_z-C_{\zest}) \Delta_c)}.
   \end{align*}
    Now, we will take advantage of the duality between the trace and operator norm (\cref{lemma:trace-norm-dual}) to correlate the distance between perturbations to the hamming distance between $z$ and $\zest$. Let $W_z = \sum_{i=1}^\ell z_i V_i$. For $z$ such that $\opnorm{W_z} \leq \kappa_\alpha \sqrt{\dims}$, we have $C_z=1$, from \cref{eq:op-norm-bound}.
   \begin{align*}
    \tracenorm{\sigma_{\zest} - \sigma_{\ptb}} &=
     \tracenorm{((1+C_{\zest}) \Delta_w + (1-C_{\zest}) \Delta_c} = \sup_{\opnorm{B} \leq 1} |\Tr[B^{\dagger} \left[(1+C_{\zest}) \Delta_w + (1-C_{\zest}) \Delta_c\right]]| \\
     &\geq \frac{1}{\kappa_\alpha \sqrt{\dims}} 
     |\Tr[W_z^{\dagger} \left[(1+C_{\zest}) \Delta_w + (1-C_{\zest}) \Delta_c\right]]| = \frac{c \eps}{\sqrt{\dims \ell}} \frac{1}{\kappa_\alpha \sqrt{\dims}} |(1+C_{\zest}) \delta_w + (1-C_{\zest}) \delta_c| \\
     &= \frac{c \eps}{\kappa_\alpha d \sqrt{\ell}} \left[(1+C_{\zest}) \delta_w + (1-C_{\zest}) \delta_c\right] \geq \frac{c \eps}{2 \kappa_\alpha \ell} \delta_w.
   \end{align*}
   Where $\delta_w = \ham{z}{\zest}$ and $\delta_c = \ell - \delta_w = \ell - \ham{z}{\zest}$.The second inequality uses: $B = \frac{W_z}{\kappa_\alpha \sqrt{\dims}}$. The reduction $\Tr[W_z^{\dagger} \Delta_w] = \delta_w$ and $\Tr[W_z^{\dagger} \Delta_c] = \delta_c$ comes directly from the orthonormality of the peturbation matrices $\{V_i\}_{i=1}^\ell$ under the inner product: $\hdotprod{A}{B} = \vvdotprod{A}{B} = \Tr[A^\dagger B]$. With the last line, we have shown the desired bound.  
   \end{proof}
    
   Since this relation to $\ham{\cdot}{\cdot}$ only occurs for a concentrated set $\mathcal{G}$, we can show that the expected hamming distance is "approximately trace distance" for sufficiently large $\dims > \frac{\ln{5}}{\alpha}$. $\sigma_{\hat{z}}$ also has to be close to $\sigma_z$ with high probability to be a sufficient estimator of $\sigma_z$, inducing a upper bound on the error probability of estimating $Z$,
   \begin{align}
       \frac{1}{\ell} \expectDistrOf{}{\delta_w} &= \frac{1}{\ell} \expectDistrOf{}{\delta_w \mid \tracenorm{\sigma_z - \sigma_{\zest}} \leq 2 \eps} \Pr[\tracenorm{\sigma_z - \sigma_{\zest}} \leq 2 \eps] \nonumber\\
       &+ \frac{1}{\ell} \expectDistrOf{}{\delta_w | \tracenorm{\sigma_z - \sigma_{\zest}} > 2 \eps} \Pr[\tracenorm{\sigma_z - \sigma_{\zest}} > 2 \eps] \nonumber\\
       &\leq \frac{1}{\ell} \expectDistrOf{}{\delta_w \mid \tracenorm{\sigma_z - \sigma_{\zest}} \leq 2 \eps} + 0.01.  \label{eq:cond-expect-tom-lower}
   \end{align}
   It is enough to upper bound the remaining expectation term by a constant. We will case on whether $z$'s lead to a approximate hamming relationship with trace distance. When $z \in \mathcal{G}$, we apply \cref{lemma:hamm-packing}
   \begin{align*}
       \frac{c \eps}{2 \kappa_\alpha \ell} \delta_w \leq \tracenorm{\sigma_z - \sigma_{\zest}} \leq 2 \eps 
       \implies  \frac{1}{\ell} \delta_w \leq \frac{4 \kappa_\alpha}{c}.
   \end{align*}
   The conditional expectation will now be bounded by a small constant for $c \geq 10 \kappa_\alpha$,
   \begin{align*}
       \frac{1}{\ell} \expectDistrOf{}{\delta_w \mid \tracenorm{\sigma_\ptb - \sigma_{\zest}} \leq 2 \eps} &\leq \Pr[z \in \mathcal{G}]\frac{1}{\ell} \expectDistrOf{}{\delta_w \mid \tracenorm{\sigma_\ptb - \sigma_{\zest}} \leq 2 \eps \land z \in \mathcal{G}} \\
       &+\Pr[z \notin \mathcal{G}] \frac{1}{\ell} \expectDistrOf{}{\delta_w \mid \tracenorm{\sigma_z - \sigma_{\zest}} \leq 2 \eps \land z \notin \mathcal{G}} \\
       &\leq \frac{4 \kappa_\alpha}{c} +  2\exp\{-\alpha d\} \cdot 1 \leq \frac{2}{5} + \frac{2}{5} = 0.40.
   \end{align*}
   Substituting this result into \cref{eq:cond-expect-tom-lower}, we have $\frac{1}{\ell} \sum_{i=1}^\ell \Pr[Z_i \neq \hat{Z_i}] = \frac{1}{\ell} \expectDistrOf{}{\delta_w} \leq 0.41$. We can then apply \cref{lem:MI-lower} to obtain the mutual information bound,
   \begin{align*}
       \frac{1}{\ell} \sum_{i=1}^\ell \mutualinfo{Z_i}{Y} \geq 1 - h\left(\frac{1}{\ell} \sum_{i=1}^\ell \Pr[Z_i \neq \hat{Z_i}]\right) \geq 1 - h\left(0.41\right) \geq \frac{1}{100}.
   \end{align*}
   The first inequality is due to the concavity of the binary entropy function $h$.
\end{proof}
\section{Lower bound for tomography with Pauli measurements}
The key to proving a tight lower bound for Pauli measurements is to design a measurement-dependent hard instance. Recall that any quantum state $\rho$ is a linear combination of Pauli observables,
\[
\rho=\frac{\eye_\dims}{\dims}+\sum_{P\in\pauliObsSet}\frac{\Tr[\rho P]}{\dims}P.
\]

Further recall the observation \cref{sec:techniques} that Pauli measurements \eqref{equ:pauli-measurement} are better at learning information about directions 
$P\in \pauliObsSet$ with a small weight and less powerful $P$ with a larger weight. 
As such, we set the basis $V_1, \ldots, V_{\dims^2-1}$ in the lower bound construction (\cref{def:perturbation}) to be the (normalized) Pauli observables, sorted in increasing order of their weights, $w(V_1)\le w(V_2)\le \ldots \le w(V_{\dims^2-1})$. Applying \eqref{equ:avg-MI-partial} in \cref{thm:avg-MI-upper} and \cref{lem:avg-MI-lower},
\begin{equation}
    \frac{1}{100}\le \frac{1}{\ell}\sum_{i=1}^{\ell}\mutualinfo{\ptb_i}{\out^\ns}\le \frac{8 \ns c^2 \eps^2}{\ell^2} \sup_{\POVM\in\povmset}\sum_{i=1}^{\ell} \vadj{V_i}\Choi_{\POVM} \vvec{V_i} +16\exp\{-\alpha\dims\}\ns c^2\eps^2.
    \label{equ:pauli-lower-inequ}
\end{equation}

We need to choose an appropriate $\ell$ and to upper bound the average mutual information. We propose to select all Pauli observables with weight at least $N-w$. Then,
\begin{equation}
    \ell = g(w)\eqdef \sum_{m=0}^w{\nqubits \choose N-m}3^{N-m}.
    \label{equ:pauli-weight-number}
\end{equation}
This is because for Pauli observables with weight $N-m$, there are $N-m$ positions we can place the Pauli operators, and for each position, there are three choices $\pauliX, \pauliY, \pauliZ$.

According to \cref{prop:perturbation-trace-distance}, we must choose $\ell\ge \dims^2/2$ to ensure that the perturbations are $\eps$ far from $\qmm$ with high probability. In other words, $g(w)/\dims^2\ge 1/2$.
\[
\frac{g(w)}{\dims^2}=\sum_{m=0}^w{\nqubits \choose N-m}\frac{3^{N-m}}{4^\nqubits}=\sum_{m=0}^{w}{\nqubits \choose m}\Paren{\frac34}^{\nqubits-m}\Paren{\frac14}^m =\probaOf{\binomial{\nqubits}{1/4}\le w}.
\]
We have the following fact about the median of binomial distributions,
\begin{fact}[\cite{kaas1980mean}]
    The median of a binomial distribution $\binomial{N}{p}$ must lie in $[\lfloor Np\rfloor, \lceil Np\rceil]$.
\end{fact}
Thus, choosing $w=\lceil N/4\rceil$ guarantees that $g(w)/\dims^2\ge 1/2$. 

Next, we compute the inner product $\vadj{V_i} \Choi_{\POVM} \vvec{V_i}$. We first need to analyze the measurement information channel of Pauli measurements.

\begin{lemma}
\label{lem:pauli-mic-eigen}
    For $P=\sigma_1\otimes\cdots\otimes \sigma_{\nqubits}\in \Sigma^{\otimes \nqubits}$, let $\Luders_P$ be the measurement information channel of the Pauli measurement $\POVM_P$. Then for all Pauli observable $Q=\sigma_1'\otimes\cdots\otimes \sigma_{\nqubits}'\in(\Sigma\cup \eye_2)^{\otimes \nqubits}$, $Q$ is an eigenvector of $\Luders_P$ and
    \[
\Luders_P(Q)=Q\indic{\forall j\in[\nqubits], \sigma_j'\in\{\sigma_j,\eye_2\}}.
    \]
    In other words, the eigenvalue of ${Q}$ is 1 when the non-identity components of $Q$ match $P$, and 0 otherwise. 
\end{lemma}
\begin{proof}
 Let $\Luders_P$ be the measurement information channel of a Pauli measurement $\POVM_{P}$. From \cref{def:mic} and \cref{equ:pauli-measurement},
\begin{align*}
    \Luders_P(\cdot)&=\sum_{x\in\{-1,1\}^{\nqubits}}\frac{M_x^{P}}{\Tr[M_x^P]}\Tr[(\cdot)M_x^P]\\
    &=\sum_{x\in\{-1,1\}^{\nqubits}}M_x^{P}\Tr[(\cdot)M_x^P].
\end{align*}
The second step is because Pauli measurement is a basis measurement. Thus each $M_x^P=\qproj{u_x^P}$ where $\{\qbit{u_x^P}\}_{x\in\{-1,1\}^{\nqubits}}$ is an orthonormal basis, and $\Tr[M_x^P]=1$.

Let $Q=\sigma_1'\otimes\cdots\otimes\sigma_{\nqubits}'$. We want to argue that $Q$ is an eigenvector of $\Luders_P$.  
\begin{align*}
\Tr[M_x^PQ]&=\Tr\left[\bigotimes_{j=1}^{\nqubits}\frac{\eye_2+x_j\sigma_j}{2}\bigotimes_{j=1}^{\nqubits}\sigma_j'\right]\\
    &=\Tr\left[\bigotimes_{j=1}^{\nqubits}\frac{\sigma_j'+x_j\sigma_j\sigma_j'}{2}\right]\\
    &=\prod_{j=1}^{\nqubits}\frac{\Tr[\sigma_j']+x_j\Tr[\sigma_j\sigma_j']}{2}\\
    &=\prod_{j=1}^{N}(\indic{\sigma_j'=\eye_2}+x_j\indic{\sigma_{j}'=\sigma_j})
\end{align*}
The final step is due to \eqref{equ:pauli-property}.
If for some $j\in[N]$, $\sigma_j'\ne \eye_2$ and $\sigma_j'\ne\sigma_j$, then
\[
\Tr[M_x^PQ]=0\implies\Luders_P(Q)=0.
\]
In this case $Q$ is an eigenvector of $\Luders_P$ with eigenvalue of 0. If otherwise,
\begin{align*}
    \Luders_P(Q)&=\sum_{x\in\{-1,1\}^{\nqubits}}M_x^P\prod_{j=1}^{N}(\indic{\sigma_j'=\eye_2}+x_j\indic{\sigma_{j}'=\sigma_j})\\
    &=\sum_{x\in\{-1, 1\}^{\nqubits}}\bigotimes_{j=1}^{\nqubits}\frac{\eye_2+x_j\sigma_j}{2}(\indic{\sigma_j'=\eye_2}+x_j\indic{\sigma_{j}'=\sigma_j})\\
    &=\bigotimes_{j=1}^{\nqubits}\sum_{x_j\in\{-1,1\}}\frac{\eye_2+x_j\sigma_j}{2}(\indic{\sigma_j'=\eye_2}+x_j\indic{\sigma_{j}'=\sigma_j})\\
    &=\bigotimes_{j=1}^{\nqubits}\Paren{\indic{\sigma_j'=\eye_2}\sum_{x_j\in\{-1,1\}}\frac{\eye_2+x_j\sigma_j}{2}+\indic{\sigma_j'=\sigma_j}\sum_{x_j\in\{-1,1\}}\frac{x_j\eye_2+\sigma_j}{2}} \\
    &=\bigotimes_{j=1}^{\nqubits}(\eye_2\indic{\sigma_j'=\eye_2}+\sigma_j\indic{\sigma_{j}'=\sigma_j})\\
    &=\bigotimes_{j=1}^{\nqubits}\sigma_j'=Q.\qedhere
\end{align*}
\end{proof}

Let $P,Q$ be defined in \cref{lem:pauli-mic-eigen} and $\Choi_P$ be the matrix form of $\Luders_P$ given a Pauli measurement $\POVM_P$. An immediate corollary is that
\begin{equation}
    \frac{1}{\dims}\vadj{Q} {\Choi}_P \vvec{Q}=\indic{\forall j\in[\nqubits], \sigma_j'\in\{\sigma_j,\eye_2\}}.
    \label{equ:pauli-mic-sum}
\end{equation}

When $V_1, \ldots, V_{\dims^2-1}$ are the normalized Pauli observables sorted in increasing order of their weights, setting $\ell=g(w)$ (the number of Pauli observables with weight at least $N-w$), we have
\begin{align*}
    \sum_{i=1}^{\ell}\vadj{V_i} {\Choi}_P \vvec{V_i}=\sum_{m=0}^{w}{\nqubits\choose m}.
\end{align*}
This is because $\vadj{V_i} {\Choi}_P \vvec{V_i}=1$ 
only when $V_i$ has non-identity components that match the ones in $P$ and 0 otherwise. 
There are only ${\nqubits\choose N-m}={\nqubits\choose m}$ of them among all $V_i$'s with weight $\nqubits-m$.

The following result gives an upper bound on the sum of binomial coefficients,
\begin{lemma}[{\cite[Lemma 16.19]{downey2012parameterized}}]
\label{lem:sum-binomial-coef}
    Let $n\ge 1$ and $0\le q\le 1/2$, then
    \[
    \sum_{i=0}^{\lfloor nq\rfloor}{n\choose i}\le 2^{n h(q)},
    \]
    where $h(q)=-q\log q -(1-q)\log(1-q)$ is the binary entropy function.
\end{lemma}

Combining with \eqref{equ:pauli-lower-inequ}\eqref{equ:pauli-weight-number}\eqref{equ:pauli-mic-sum}, setting $w=\lceil N/4\rceil$,
\begin{align*}
    \frac{1}{100}&\le \frac{8 \ns c^2 \eps^2}{\ell^2} \sup_{P\in\Sigma^{\otimes\nqubits}}\sum_{i=1}^{\ell}\vadj{V_i} {\Choi}_P \vvec{V_i} +16\exp\{-\alpha\dims\}\ns c^2\eps^2\\
    &=8nc^2\eps^2\Paren{\frac{\sum_{m=0}^{w}{\nqubits\choose m}}{g(w)^2}+2\exp\{-\alpha\dims\}}\\
    &\le 8nc^2\eps^2\Paren{\frac{2\cdot 2^{\nqubits h(1/4)}}{\dims^4/4}+2\exp\{-\alpha\dims\}}\\
    &\le 16\ns \cd^2\eps^2\Paren{4\cdot 2^{(h(1/4)-4)\nqubits}+\exp\{-\alpha2^{\nqubits}\}}.
\end{align*}
When $N\ge 10$, the second term is negligible. Rearranging the terms, we must have
\[
\ns = \bigOmega{\frac{2^{(4-h(1/4))\nqubits}}{\eps^2}}.
\]
Finally, noting that $2^{4-h(1/4)}\ge 9.118$ completes the proof.

\section{Upper bound for Pauli measurements}
\label{sec:pauli-upper}
This section starts with an observation about Pauli measurements, which is common knowledge for quantum information experimentalists. Then, we employ this observation to improve previous results about quantum state tomography using Pauli measurements.

\subsection{ An Observation about Pauli Measurements}
When we measure an element of the Pauli group, for instance, $\sigma_X\otimes \sigma_Y$, on a two-qubit state $\rho$, the outcome is a sample from a $4$-dimensional probability distribution, says $(p_{00},p_{01},p_{10},p_{11})$, such that
\begin{align*}
\tr(\rho(\sigma_X\otimes \sigma_Y))=p_{00}-p_{01}-p_{10}+p_{11}.
\end{align*}

One can easily observe that
\begin{align*}
\Tr[\rho(\sigma_X\otimes \sigma_I)]=p_{00}+p_{01}-p_{10}-p_{11},\\
\Tr[\rho(\sigma_I\otimes \sigma_Y)]=p_{00}-p_{01}+p_{10}-p_{11},\\
\Tr[\rho(\sigma_I\otimes \sigma_I)]=p_{00}+p_{01}+p_{10}+p_{11}.
\end{align*}

In other words, measuring $XY$, we obtained a sample of $\sigma_X\sigma_I$, a sample of $\sigma_I\sigma_Y$, and a sample of $\sigma_I\sigma_I$. 

For a general $n$-qubit system, we have the following observation.
\begin{observation}
For any $P=P_1\otimes P_2\otimes\cdots\otimes P_n\in\{\sigma_X,\sigma_Y,\sigma_Z\}^{\otimes N}$, the measurement result of performing measurement $P_i$ on the $i$-th qubit is an $N$-bit string $s$. One can interpret the measurement result of performing $Q_i\in\{\sigma_I,\sigma_X,\sigma_Y,\sigma_Z\}$ on the $i$-th qubit if $Q_i=P_i$ or $Q_i=\sigma_I$. We call those $Q=Q_1\otimes Q_2\otimes\cdots\otimes Q_N$'s correspond to $P$.
\end{observation}

\subsection{Algorithm and error analysis} 
Our measurement scheme is as follows: For any $\eps>0$, fix an integer $m$.
\begin{enumerate}
    \item  For any $P\in\{\sigma_X,\sigma_Y,\sigma_Z\}^{\otimes N}$, one performs $m$ times $P$ on $\rho$, and records the $m$ samples of the $2^N$ dimensional outcome distribution.

According to the key observation, this measurement scheme provides $m\cdot 3^{N-w}$ samples of the expectation $\tr(\rho P)$, say, $\frac{\mu_P}{m\cdot 3^{N-w}}$, for each Pauli operator $P\in \{\sigma_I,\sigma_X,\sigma_Y,\sigma_Z\}^{\otimes N}$ with weight $w$, where $-m\cdot 3^{N-w}\leq\mu_P\leq m\cdot 3^{N-w}$.

\item Output 
\begin{align*}
\sigma=\sum_P \frac{\mu_P}{m\cdot 3^{N-w}\cdot 2^N} P.
\end{align*}
\end{enumerate}

Using this scheme, we obtained $m\cdot 3^N$ independent samples,
\begin{align*}
X_1,X_2,\cdots, X_{m\cdot 3^N}.
\end{align*}
Each  $X_i$ is an $N$-bit string recording outcomes on all qubits (using bit 0 to denote the +1 eigenvalue and bit 1 to denote the -1 eigenvalue of the measured Pauli operator).   
Given that each operator is measured $m$ times, specifically, we assign that $X_1,X_2,\cdots,X_{m}$ correspond to the measurement $\sigma_X^{\otimes N}$, $X_{m+1},X_{m+2}$, and $\cdots,X_{2m}$ corresponds to the measurement $\sigma_X^{\otimes N-1}\otimes \sigma_Y$, \dots, and until $\sigma_Z^{\otimes N}$.

We observe that for any $P$ of weight $w$,
$\mu_P=\sum_{j=0}^{m\cdot 3^{N-w}-1} Z_j$,
where $Z_j$ are independent samples from the distribution $Z$
\begin{align*}
\mathrm{Pr}(Z=1)=\frac{1+\tr (\rho P)}{2}, \;
\mathrm{Pr}(Z=-1)=\frac{1-\tr (\rho P)}{2}.
\end{align*}
We have
\begin{align*}
\expectDistrOf{}{Z}&=\tr (\rho P), \expectDistrOf{}{Z^2}=1, \\
\expectDistrOf{}{\mu_P}&=m\cdot 3^{N-w}\cdot\tr (\rho P), \\
\expectDistrOf{}{\mu_P^2}&=\expectDistrOf{}{\mu_P}^2+\Var[\mu_P] \\
&=\expectDistrOf{}{\mu_P}^2+m\cdot 3^{N-w}\Var[Z] \\
&=m^2\cdot 9^{N-w}\cdot \tr^2 (\rho P)+m\cdot 3^{N-w}(1-\tr^2 (\rho P)).
\end{align*}
Thus, we can verify that
\begin{align*}
\expectDistrOf{}{\sigma}=\rho,
\end{align*}
where the expectation is taken over the probabilistic distribution according to the measurements.

For convenience, we define the function $f:X_1\times X_2\times \cdots\times X_{m\cdot 3^N}\mapsto \mathbb{R}$
\begin{align*}
f(\sigma)=\hsnorm{\rho-\sigma}=\sqrt{{\rm Tr}[(\rho-\sigma)^\dagger (\rho-\sigma)]}.
\end{align*}
Note that we can write the unknown state $\rho$ as
\begin{align*}
\rho=\sum_P \frac{\alpha_P}{2^N} P.
\end{align*}
According to Cauchy-Schwarz and Jensen's inequality, we have
\begin{align*}
&\expectDistrOf{}{f(\sigma)} \leq\sqrt{\expectDistrOf{}{f(\sigma)^2}} = \sqrt{\expectDistrOf{}{\tr\rho^2-2\tr \rho\sigma+\tr\sigma^2}}\\
=& \sqrt{\expectDistrOf{}{\tr\sigma^2-\tr\rho^2}}
= \sqrt{\frac{1}{2^N}\sum_P \expectDistrOf{}{\frac{\mu_P^2}{m^2\cdot 9^{N-w_P}}-\alpha_P^2}} \\
=& \sqrt{\frac{1}{2^N}\sum_P \left(\frac{m^2\cdot 9^{N-w_P}\cdot \alpha_P^2+m\cdot 3^{N-w_P}(1-\alpha_P^2)}{m^2\cdot 9^{N-w_P}}-\alpha_P^2\right)} \\
=&\sqrt{\frac{1}{m\cdot 2^N}\cdot{\sum_P \frac{1-\alpha_P^2}{3^{N-w_P}}}} < \sqrt{\frac{1}{m\cdot 2^N}\cdot{\sum_P \frac{1}{3^{N-w_P}}}}\\
=& \sqrt{\frac{1}{m\cdot 2^N}\cdot{\sum_{w_P=0}^N \frac{1}{3^{N-w_P}}{{N}\choose{w_P}}3^{w_P}}} = \sqrt{\frac{1}{m\cdot 6^N}\cdot (1+9)^N} \\
=& \sqrt{\frac{5^{\nqubits}}{m\cdot 3^N}}.
\end{align*}

For any sample $X_i$ corresponding to $P\in\{\sigma_X,\sigma_Y,\sigma_Z\}^{\otimes N}$, if only $X_i$ is changed, $\mu_Q$ would be changed only for those $Q\in\{\sigma_I,\sigma_X,\sigma_Y,\sigma_Z\}^{\otimes N}$ where
$Q$ is obtained by replacing some $\{\sigma_X,\sigma_Y,\sigma_Z\}$'s of $P$ by $\sigma_I$.
Moreover, the resultant value of $\mu_Q$ would change by two at most.
According to the triangle inequality, $f$ would change at most 
\begin{align*}
\hsnorm{\sum_Q  \frac{\Delta\mu_Q}{m\cdot 3^{N-w_Q}\cdot 2^N} Q}&=\sqrt{\sum_Q \frac{\Delta\mu_Q^2}{m^2\cdot 9^{N-w_Q}\cdot 2^N}}\\
&\leq \sqrt{\sum_Q \frac{2^2}{m^2\cdot 9^{N-w_Q}\cdot 2^N}} \\
&= \sqrt{\sum_{w_Q=0}^N \frac{2^2}{m^2\cdot 9^{N-w_Q}\cdot 2^N}{{N}\choose{w_Q}}} = \frac{2\cdot \sqrt{5}^N}{m\cdot 3^N},
\end{align*}
where $Q$ ranges over all Paulis which correspond to $P$'s, and $\Delta\mu_Q$ denotes the difference of $\mu_Q$ when $X_i$ is changed.

We use McDiarmid's inequality to bound the probability of success.
\begin{lemma}\label{mc}
Consider independent random variables ${\displaystyle X_{1},X_{2},\dots X_{n}}$ on probability space $ {\displaystyle (\Omega ,{\mathcal {F}},{\text{P}})}$ where ${\displaystyle X_{i}\in {\mathcal {X}}_{i}}$ for all ${\displaystyle i}$ and a mapping ${\displaystyle f:{\mathcal {X}}_{1}\times {\mathcal {X}}_{2}\times \cdots \times {\mathcal {X}}_{n}\rightarrow \mathbb {R} }$. Assume there exist constant $ {\displaystyle c_{1},c_{2},\dots ,c_{n}} $ such that for all $ {\displaystyle i}$,
\begin{align}{\displaystyle {\underset {x_{1},\cdots ,x_{i-1},x_{i},x_{i}',x_{i+1},\cdots ,x_{n}}{\sup }}|f(x_{1},\dots ,x_{i-1},x_{i},x_{i+1},\cdots ,x_{n})-f(x_{1},\dots ,x_{i-1},x_{i}',x_{i+1},\cdots ,x_{n})|\leq c_{i}.} 
\end{align}
In other words, changing the value of the ${\displaystyle i}$-th coordinate ${\displaystyle x_{i}}$ changes the value of ${\displaystyle f}$ by at most ${\displaystyle c_{i}}$. Then, for any ${\displaystyle \epsilon >0}$,
\begin{align} 
{\displaystyle {\mathrm{Pr}}(f(X_{1},X_{2},\cdots ,X_{n})-\expectDistrOf{}{f(X_{1},X_{2},\cdots ,X_{n})}\geq \epsilon )\leq \exp \left(-{\frac {2\epsilon ^{2}}{\sum _{i=1}^{n}c_{i}^{2}}}\right)} .
\end{align}
\end{lemma}
We only consider $\delta<1/3$, then $\log(1/\delta)>1$.
For any $\eps'>0$, by choosing $m=(3+2\sqrt{2})\cdot\frac{5^N\log\frac{1}{\delta}}{3^N\cdot \eps'^2}$, we have $\expectDistrOf{}{f(\sigma)}< (\sqrt{2}-1){\eps'}$.
Therefore,
\begin{align*}
\mathrm{Pr}(f(\sigma) > \eps') &< \mathrm{Pr}(f(\sigma)-\expectDistrOf{}{f(\sigma)}>(2-\sqrt{2}){\eps'}) \\
&<\exp(-\frac{(12-8\sqrt{2})\cdot \eps'^2}{4\cdot \frac{5^N}{m^2\cdot 9^N}\cdot m\cdot 3^N}) \\
=& \exp(-\frac{m\cdot(3-2\sqrt{2})\cdot 3^N\cdot\eps'^2}{5^N}) <\delta, 
\end{align*}
where the inequality is by \cref{mc}.

For a general quantum state and $\eps>0$, we let $\eps'=\frac{\eps}{\sqrt{2^N}}$, and know that
$||\rho-\sigma||_1>\eps$ implies $\hsnorm{\rho-\sigma}>\eps'$. Therefore,
\begin{align*}
\mathrm{Pr}(||\rho-\sigma||_1>\eps) &\leq \mathrm{Pr}(\hsnorm{\rho-\sigma}>\eps')=\mathrm{Pr}(f>\eps').
\end{align*}

The total number of used copies is 
\begin{align*}
\ns = m\cdot 3^N=(3+2\sqrt{2})\cdot\frac{10^N\log\frac{1}{\delta}}{\eps^2}.
\end{align*}

\newcommand{\MUB}{\POVM_{MUB}}
\section{Upper bound for tomography with finite outcomes}
\label{sec:finite-upper}
We will show the tightness of the adaptive tomography bounds for k-outcome POVMs by modifying the Projected Least Squares Method (PLS)~\cite{guctua2020fast} to work with k-outcome POVMs. We present these adjustments for the case when $k = d$ and $k < d$. As a result, we will have the first upper and lower bounds for adaptive tomography for k-outcome measurements, where the upper bound is achieved with non-adaptive algorithms. The key component is reducing the $\MUB$ to a k outcome measurement,
$$\MUB \eqdef \left\{\frac{1}{d+1} \ket{\psi_x^k} \bra{\psi_x^k}\right\}_{k \in [d+1], x \in [d]},$$
where each fixed $k$ corresponds to each one of the Maximally mutually unbiased bases. The reduction will follow similarly to~\cite{liu2024restricted}.
\subsection{Algorithm for $k=d$} \label{sub-keqd}
Measuring with $\MUB$ acts as a uniform sampling $i \sim Unif([d+1])$ to select one of the MUB and measuring with the POVM described by $\left\{\ket{\psi_x^i} \bra{\psi_x^i}\right\}_{x \in [d]}$. So, we can split each of the MUB bases across the multiple copies and uniformly sample amongst them to replicate the outcome distribution of measuring with $\MUB$. 

\begin{algorithm}
\caption{Finite Outcome Tomography for $k = d$}\label{alg:tom-keqd}
\hspace*{0.1cm} \textbf{Input:} $n$ copies of state $\rho$ \\
\hspace*{0.1cm} \textbf{Output:} Estimate $\hat{\rho} \in \mathcal{C}^{d \times d}$
\begin{algorithmic}
\State Divide $\MUB$ into $d+1$ groups of d-outcome measurements $\mathcal{M}_j := \left\{\ket{\psi_x^j} \bra{\psi_x^j}\right\}_{x \in [d]}$.
\State Divide $n$ copies into $d+1$ equally sized groups, each group has $n_0 = n/(d+1)$ copies.
\For {$j = 1, ..., d+1$}
\State For group $j$, apply $\mathcal{M}_j$. Let the outcomes be $x_1^{(j)}, ..., x_{n_0}^{(j)}$.
\EndFor
\State Generate n/2 i.i.d samples from $Unif([d+1])$ and  let $m_j$ be the number of times $j$ appears.
\State Let $x = (x_1, ..., x_{d+1})$ where $x_j = (x_1^{(j)}, ..., x_{\min\{n_0, m_j\}}^{(j)}$)
\State From $x$, obtain empirical frequencies $F = (f_1, ..., f_{d (d+1)})$ by obtaining group specific frequencies of each $x_i$ and concatenating the frequency vectors together.
\State \Return $\hat{\rho} = PLS(F)$
\end{algorithmic}
\end{algorithm}

For the analysis, we will use the multiplicative Chernoff Bound for sums of i.i.d random variables.
\begin{lemma}[Multiplicative Chernoff Bound]
    \label{mult_chernoff}
   Let $X_1, ..., X_n$ be i.i.d with $\expectDistrOf{}{X_1} = \mu$. Then, 
   $$\Pr\left[\sum_{i}^n X_i \geq n(1+\alpha)\mu\right] \leq \exp{\left\{-\frac{n \alpha^2 \mu}{2 + \alpha}\right\}}\;, \alpha > 0$$ 
   $$\Pr\left[\sum_{i}^n X_i \geq n(1-\alpha)\mu\right] \leq \exp{\left\{-\frac{n \alpha^2 \mu}{2}\right\}}\;, \alpha \in (0,1)$$ 
\end{lemma}
\begin{theorem} \label{thm_keqd}
    For $k=d$, Algorithm \cref{alg:tom-keqd} will give estimate $\hat{\rho}$ such that $\Pr[\tracenorm{\hat{\rho} - \rho} \leq \eps] \geq \frac{2}{3}$ with $n = O\left( \frac{d^3 \log d}{\eps^2}\right)$
\end{theorem}
\begin{proof}
Notice that each sample made will follow the outcome distribution of applying $\MUB$ to a single copy of $\rho$. 
Given $n$ copies, it will be shown that $\frac{n}{2}$ such samples will be made with sufficiently high probability. This is when $m_j \leq n_j$ for all $j \in [d+1]$. 
Using~\cref{mult_chernoff}, on the $m_j \sim Bin(\frac{n}{2}, \frac{1}{d+1})$, which is sum of $Y_1,...,Y_{\frac{n}{2}} \sim Bern(\frac{1}{d+1})$, $\mu = \expectDistrOf{}{Y_1} = \frac{1}{d+1}$,
\begin{align*}
    \Pr\left[m_j > n_j\right] = \Pr\left[\sum_{i=1}^\frac{n}{2} Y_i > 2 n \mu \right] \leq \exp{\left\{-\frac{n}{6(d+1)}\right\}}.
\end{align*}
Furthermore, by union bound,
\begin{align*}
    \Pr\left[\exists_j m_j > n_j\right] \leq \sum_{j=1}^{d+1} \Pr\left[m_j > n_j\right] \leq (d+1) \exp{\left\{-\frac{n}{6(d+1)}\right\}}.
\end{align*}
From previous work \cite{guctua2020fast}, we have the following guarantee on the estimation error using the PLS method using the outcome of $\MUB$ measurements,
\begin{align*}
    \Pr\left[\tracenorm{\hat{\rho}_n - \rho} \geq \eps \right] \leq d \exp{\left\{-\frac{n \eps^2}{86 d^3}\right\}}.
\end{align*}
With the algorithm, we can bound the probability of the estimate not being optimal,
\begin{align*}
\Pr\left[\tracenorm{\hat{\rho} - \rho} \geq \eps \right] &\leq \Pr\left[\exists_j m_j > n_j \lor \tracenorm{\hat{\rho}_{n/2} - \rho} \geq \eps\right] \\
&\leq (d+1) \exp{\left\{-\frac{n}{6(d+1)}\right\}} + d \exp{\left\{-\frac{n \eps^2}{172 d^3}\right\}}.
\end{align*}
With $d \geq 16$ and $n = \frac{172d^3 \ln{200d}}{\eps^2}  = O\left(\frac{d^3 \log{d}}{e^2} \right)$, we will have $\Pr\left[\tracenorm{\hat{\rho} - \rho} \le \eps \right] \geq \frac{99}{100}$
\end{proof}
\subsection{Algorithm for $k < d$}
For $\ab<\dims$, it is helpful to think of the problem as follows: there are $\ns$ players, each of whom holds a copy of $\rho$, but can only send $\log\ab$ classical bits to a central server that collects the messages and learn about the state. 

The idea is then to simulate each $\dims$-outcome POVM using only $\log\ab$ bits for each player. 
Using results from \cite{ACT:19:IT2}, the number of players (or copies) required to simulate the original $\dims$-outcome POVM is roughly $O(\dims/\ab)$, and thus we have a $O(\dims/\ab)$ factor blow up in the sample complexity compared to $\dims$-outcome measurements.

\begin{definition} [$\eta$-Simulation]
\label{def-simulate}
We are given $n$ players each with i.i.d sample from an unknown distribution $\p \in \Delta_d$. Each player can only send $w$ bits to the server. The server can then perform a $\eta$-simulation where $\hat{X} = [d] \cup \{\perp\}$.
\begin{equation}
    \Pr[\hat{X} = x \mid \hat{X} \neq \perp] = p_x, \; \Pr[\hat{X} = \perp] \leq \eta
\end{equation}

It can be shown that there exists an algorithm that can perform a $\eta$ simulation with $O(\dims/\ab)$ players,

\end{definition}
\begin{theorem}[\cite{ACT:19:IT2}, Theorem IV.5]
\label{thm:simulation}
   For every $\eta \in (0,1)$, there exists an algorithm that $\eta$ simulates $p \in \Delta_d$ using 
   \begin{equation}
       M = 40 \left\lceil \log{\frac{1}{\eta}} \right\rceil \left\lceil \frac{d}{2^w - 1} \right\rceil
   \end{equation}
   players from the setting in \cref{def-simulate}. The algorithm only requires private randomness for each player.
\end{theorem}
Therefore, for each MUB measurement $\POVM$ we assign $M=O(\dims/\ab)$ players. 
Each player applies $\POVM$ to $\rho$ and compresses the outcome to $\log\ab$ bits using the simulation algorithm in \cref{thm:simulation}. 
This process is a valid $\ab$-outcome POVM. 
The server then can use the $M=O(\dims/\ab)$ messages to simulate the outcome of $\POVM$ applied to $\rho$.
From \cref{thm_keqd}, we need $\tildeO{\dims^3/\eps^2}$ simulated samples, and thus the total number of copies required to simulate those samples is $M=O(\dims/\ab)$ times more. The detailed proof is given in \cref{thm:k-outcome-tomography}.


\begin{theorem}
\label{thm:k-outcome-tomography}
For $\ab < \dims$, Algorithm 1 with distributed simulation will give estimate $\hat{\rho}$ such that $Pr[\tracenorm{\rho-\hat{\rho}} \le \eps ]\ge 0.99$ with  $\ns = \bigO{\frac{\dims^4\log\dims}{\ab\eps^2}}$.
\end{theorem}
\begin{proof}
   The proof will follow the same steps as \cref{thm_keqd}, but also considering $n_j \sim Bin(1-\eta,n_0/M)$ taking the role of $n_0$. Since $n_j$ and $m_j$ are both Binomial random variables, it is enough to say that $m_j$ and $n_j$ are on the opposite sides of a mean threshold with exponentially decreasing probability. Denote $\hat{n}_0 = \frac{n_0}{M}$ and $\hat{n} = \frac{n}{M}$,
   \begin{align*}
       \Pr\left[m_j \leq n_j\right] &\geq \Pr\left[m_j \leq \frac{3}{4} \hat{n}_0 \land n_j \geq \frac{3}{4} \hat{n}_0\right] \\
       \Pr\left[m_j > n_j\right] &< \Pr\left[m_j > \frac{3}{4} \hat{n}_0 \lor n_j < \frac{3}{4} \hat{n}_0\right].
   \end{align*}
   We will now bound each of the above union events with exponentially decreasing probability. We will apply \cref{mult_chernoff} on $m_j \sim Bin(\hat{n}/2, \frac{1}{d+1}), \expectDistrOf{}{m_j} = \hat{n}_0/2$ with $\alpha = 1/2$,
   \begin{align*}
       \Pr\left[m_j > \frac{3}{4} \hat{n}_0\right] \leq \exp{\left\{- \frac{\hat{n}}{10(d+1)}\right\}}.
   \end{align*}
   Now, we will apply \cref{mult_chernoff} once more for $n_j \sim Bin(1-\eta,\hat{n}_0), \expectDistrOf{}{n_j} = (1-\eta) \hat{n}_0$ and $\alpha = \frac{1/4 + \eta}{\eta} \geq \frac{1}{4}$,
   \begin{align*}
          \Pr\left[n_j > \frac{3}{4} \hat{n}_0\right] \leq \exp{\left\{- \frac{\hat{n}}{32(d+1)}\right\}}.
   \end{align*}
   Thus, 
   \begin{align*}
       \Pr\left[m_j > n_j \right] &\leq \Pr\left[m_j > \frac{3}{4} \hat{n}_0 \right] + \Pr\left[n_j < \frac{3}{4} \hat{n}_0\right] \\
       &\leq 2 \exp{\left\{- \frac{\hat{n}}{32(d+1)}\right\}}.
   \end{align*}
   By union bound,
   \begin{align*}
       \Pr\left[\exists_j m_j > n_j \right] \leq (d+1) \exp{\left\{- \frac{\hat{n}}{32(d+1)}\right\}}.
   \end{align*}
   Now we will repeat the argument from \cref{sub-keqd} for plugging in the samples into the PLS estimator,
   \begin{align*}
       \Pr\left[\tracenorm{\hat{\rho} - \rho} > \eps \right] &\leq \Pr\left[\exists_j m_j > n_j \cup \tracenorm{\hat{\rho}_{\hat{n}/2} - \rho} > \eps\right] \\
       &\leq (d+1) \exp{\left\{-\frac{\hat{n}}{32(d+1)}\right\}} + d \exp{\left\{-\frac{\hat{n} \eps^2}{172 d^3}\right\}}.
   \end{align*}
With $d \geq 16$ and $\hat{n} = \frac{172d^3 \ln{200d}}{\eps^2}$, we will have $\Pr\left[\tracenorm{\hat{\rho} - \rho} \le \eps \right] \geq \frac{99}{100}$. With $w = \log k$ and $\eta = 0.01$, we have that $\hat{n} = \Theta(\frac{k}{d}) \cdot n$, so $n = O(\frac{d^4 \log{d}}{k \eps^2} )$.
\end{proof}
Thus, the upper bound can be compactly written as $O\left(\frac{\dims^4 \log{d}}{\eps^2\min\{\ab, \dims\}}\right)$, combining the $k=d$ and $k < d$ cases. With \cref{cor:finite-lower}, we have proven \cref{thm:nearly-tight-finite-out}.

\begin{remark}
    We note that running the distributed simulation with $\log\ab$ bits requires first obtaining the $\dims$ outcomes for each qubit. Thus, the algorithm is more relevant in the distributed setting as described in this section. Nevertheless, the compression step for each copy defines a valid $\ab$-outcome measurement and thus proves that our lower bound in~\cref{cor:finite-lower} is tight.
\end{remark}

\section*{Acknowledgements}

JA and YL were partially supported by NSF award 1846300 (CAREER), NSF CCF-1815893. AD was supported by a Cornell University Graduate Fellowship. YL was also supported by a Rice University Chairman postdoctoral fellowship. NY was supported by DARPA  SciFy Award 102828.

\bibliography{refs}
\bibliographystyle{alpha}

\end{document}